\documentclass[twocolumn]{bmcart}

\usepackage{graphicx}
\usepackage{amsmath}
\usepackage{bbm}
\usepackage{cite}
\usepackage{amssymb}
\usepackage{amsthm}
\usepackage{mathtools,bm}
\usepackage{mathrsfs}
\usepackage[utf8]{inputenc} 
\usepackage{changes}
\usepackage{comment}
\usepackage{svg}
\usepackage[font=small]{caption}
\usepackage{subcaption}
\usepackage{enumitem}
\usepackage{adjustbox}
\usepackage{nomencl}
\usepackage{etoolbox}
\renewcommand\nomgroup[1]{%
  \item[\bfseries
  \ifstrequal{#1}{A}{Notations}{%
  \ifstrequal{#1}{B}{Abbreviations}{{}}}%
]}
\makenomenclature

\newcommand{\qtilde}[0]{\tilde{q}}

\newcommand{\XX}{\mbox{\tiny \it $X\!X'$}}

\theoremstyle{plain}
\newtheorem{thm}{Theorem}
\newtheorem{lem}{Lemma} 
\newtheorem{prop}{Proposition}
\newtheorem{cor}{Corollary}
\newtheorem{rem}{Remark}
\newtheorem{defn}{Definition}

\allowdisplaybreaks
\begin{document}

\begin{frontmatter}

\begin{fmbox}
\dochead{Research}


\title{Optimum Noise Mechanism for Differentially Private Queries in Discrete Finite Sets}


\author[
  addressref={aff1},                   
  email={sachinkadam@skku.edu}   
]{\fnm{Sachin} \snm{Kadam}}
\author[
  addressref={aff2},                   
  email={as337@cornell.edu}   
]{\fnm{Anna} \snm{Scaglione}}
\author[
  addressref={aff2},                   
  email={nr337@cornell.edu}   
]{\fnm{Nikhil} \snm{Ravi}}
\author[
  addressref={aff3},                   
  email={sppeisert@lbl.gov}   
]{\fnm{Sean} \snm{Peisert}}
\author[
  addressref={aff4},                   
  email={brent@kevalaanalytics.com}   
]{\fnm{Brent} \snm{Lunghino}}
\author[
  addressref={aff4},                   
  email={aram@kevalaanalytics.com}   
]{\fnm{Aram} \snm{Shumavon}}

\address[id=aff1]{
  \orgdiv{Technology Innovation Hub (TIH) Foundation for IoT and IoE},             
  \orgname{IIT Bombay campus},          
  \city{Mumbai},                              
  \cny{India}                                    
}
\address[id=aff2]{%
  \orgdiv{Department of Electrical and Computer Engineering},
  \orgname{Cornell Tech,  Cornell University},
  \city{New York},
  \cny{USA}
}
\address[id=aff3]{%
  \orgname{Lawrence Berkeley National Laboratory},
  \city{Berkeley, CA},
  \cny{USA}
}
\address[id=aff4]{%
  \orgname{Kevala, Inc.},
  \city{San Francisco, CA},
  \cny{USA}
}




\begin{abstractbox}

\begin{abstract} 
{The Differential Privacy (DP) literature often centers on meeting privacy constraints by introducing noise to the query, typically using a pre-specified parametric distribution model with one or two degrees of freedom. However, this emphasis tends to neglect the crucial considerations of response accuracy and utility, especially in the context of categorical or discrete numerical database queries, where the parameters defining the noise distribution are finite and could be chosen optimally. This paper addresses this gap by introducing a novel framework for designing an optimal noise Probability Mass Function (PMF) tailored to discrete and finite query sets.
Our approach considers the modulo summation of random noise as the DP mechanism, aiming to present a tractable solution that not only satisfies privacy constraints but also minimizes query distortion. Unlike existing approaches focused solely on meeting privacy constraints, our framework seeks to optimize the noise distribution under an arbitrary $(\epsilon, \delta)$ constraint, thereby enhancing the accuracy and utility of the response. We demonstrate that the optimal PMF can be obtained through solving a Mixed-Integer Linear Program (MILP). Additionally, closed-form solutions for the optimal PMF are provided, minimizing the probability of error for two specific cases. Numerical experiments highlight the superior performance of our proposed optimal mechanisms compared to state-of-the-art methods.
This paper contributes to the DP literature by presenting a clear and systematic approach to designing noise mechanisms that not only satisfy privacy requirements but also optimize query distortion. 
The framework introduced here opens avenues for improved privacy-preserving database queries, offering significant enhancements in response accuracy and utility.}
\end{abstract}


\begin{keyword}
\kwd{Differential Privacy} 
\kwd{Optimum Noise Mechanism} 
\kwd{Discrete Queries} 
\kwd{MILP} 
\kwd{Error Rate}
\end{keyword}


\end{abstractbox}
\end{fmbox}
\end{frontmatter}



\nomenclature[A]{$\mathbb{N}$}{Set of natural numbers including zero}
\nomenclature[A]{$\mathbb{N}^{+}$}{Set of natural numbers excluding zero}
\nomenclature[A]{$\mathbb{Z}$}{Set of integers}
\nomenclature[A]{$\mathbb{R}$}{Set of real numbers}
\nomenclature[A]{$[n]$}{$\{0, 1, \ldots, n\}$, $n\in \mathbb{N}$}
\nomenclature[A]{$[n]_{+}$}{$\{1, \ldots, n\}$, $n\in \mathbb{N}^{+}$}
\nomenclature[A]{$\lfloor k \rfloor$}{Floor function of $k$}
\nomenclature[A]{$\lceil k \rceil$}{Ceil function of $k$}
\nomenclature[A]{$|\mathcal{A}|$}{Cardinality of set $\mathcal{A}$}
\nomenclature[A]{$\mathcal{A}\setminus\mathcal{B}$}{Set difference between sets $\mathcal{A}$ and $\mathcal{B}$}
\nomenclature[A]{$\mathcal{X}$}{Database}
\nomenclature[A]{$X$}{Data point(s) in database $\mathcal{X}$}
\nomenclature[A]{$X'$}{A neighboring dataset of $X$ differing by a data record}
\nomenclature[A]{$Pr(\mathcal{A})$}{The probability of occurrence of an event $\mathcal{A}$}
\nomenclature[A]{$\mathbb{Q}$}{Query function}
\nomenclature[A]{$\mathcal{Q}$}{Discrete set of query answers}
\nomenclature[A]{$q$}{A query answer}
\nomenclature[A]{$\tilde{q}$}{The randomized query answer}
\nomenclature[A]{$f_A$}{The PMF of a random variable $A$}
\nomenclature[A]{$f^{\star}_A$}{The optimum noise PMF of random variable $A$}
\nomenclature[A]{$f^{\star}_{\infty}$}{The optimum noise PMF as $n \rightarrow \infty$}
\nomenclature[A]{$f_{(k)}$}{The $(k+1)$\textsuperscript{th} largest probability mass}
\nomenclature[A]{$\epsilon$}{DP budget}
\nomenclature[A]{$\delta$}{Probability of information being leaked}
\nomenclature[A]{$L_{\XX}$}{Privacy loss function (eq. \ref{Eq:PrivacyLoss})}
\nomenclature[A]{${\mathcal X}^{(1)}_{X}$}{Neighborhood set of dataset $X$}
\nomenclature[A]{$\eta$}{Discrete noise random variable}
\nomenclature[A]{$\mu_{\XX}$}{Query distance between $X$ and its distance one neighbor $X'$}
\nomenclature[A]{$u_{\XX}$}{Indicator function which is active when the privacy leakage function is more than the privacy budget, $\epsilon$}
\nomenclature[A]{$\rho_{\eta}$}{Distortion caused by noise value $\eta$}
\nomenclature[A]{$\mathbb{E}[A]$}{Expectation of a random variable $A$}
\nomenclature[A]{$\rho^{\text{ER}}$}{Error rate}
\nomenclature[A]{$\rho^{\star}$}{Optimum error rate}
\nomenclature[A]{$\rho^{\star}_{\infty}$}{Optimum error rate as $n \rightarrow \infty$}
\nomenclature[A]{$\rho^{\text{MSE}}$}{Mean Squared Error}
\nomenclature[A]{$y_{\eta}$}{Auxiliary variable equal to $f(\eta)u_{\XX}(\eta)$}
\nomenclature[A]{$\mathcal{M}$}{Set of query distances, $\forall$ $X\in\mathcal{X}$ and $X'\in \mathcal{X}_X^{(1)}$}
\nomenclature[A]{$\hat{\mu}$}{The constant query distance in the Single Distance case}
\nomenclature[A]{$\bar{\mu}$}{The maximum query distance in the Bounded Difference case}
\nomenclature[A]{$\bar{\mu}$}{The maximum query distance in the Bounded Difference case}
\nomenclature[A]{$\phi_i^k$}{$i^{th}$ step height in the flat region of the PMF of the optimal noise mechanism for BD neighbourhood case}
\nomenclature[A]{$\psi_i^k(\delta)$}{$i^{th}$ step height in the linear region of the PMF of the optimal noise mechanism for BD neighbourhood case}
\nomenclature[A]{$\underline{\delta}^{\epsilon}_k$}{The instance of $\delta$ at which $\rho^{\star}(\delta,\epsilon)$ changes from $k^{th}$ linear region to $k^{th}$ flat region}
\nomenclature[A]{$\overline{\delta}^{\epsilon}_k$}{The instance of $\delta$ at which $\rho^{\star}(\delta,\epsilon)$ changes from $k^{th}$ flat region to $(k+1)^{th}$ linear region}
\nomenclature[A]{$r$}{Reminder when $n$ is divided by $\overline{\mu}$}
\nomenclature[A]{$b$}{Quotient when $n$ is divided by $\overline{\mu}$}
\nomenclature[A]{$N_{\mu}$}{Ratio between $n+1$ and gcd(($n+1$), $\mu$)}
\nomenclature[A]{$F_A$}{The cumulative distribution function (CDF) of random variable $A$}
\nomenclature[A]{$\alpha$}{Parameter of Geometric distribution used in the paper}
\nomenclature[A]{$\beta'$}{Parameter of Gumbel distribution used in the paper}
\nomenclature[B]{AMI}{ Advanced Metering Infrastructure}
\nomenclature[B]{BD}{Bounded Distance}
\nomenclature[B]{DP}{Differential Privacy}
\nomenclature[B]{ER}{Error Rate}
\nomenclature[B]{LP}{Linear Program}
\nomenclature[B]{MSE}{ Mean Squared Error}
\nomenclature[B]{PDP}{ Probabilistic Differential Privacy}
\nomenclature[B]{PMF}{ Probability Mass Function}
\nomenclature[B]{SD}{Single Distance}
\nomenclature[B]{GCD}{Greatest Common Divisor}
\nomenclature[B]{MILP}{ Mixed Integer Linear Program}
\nomenclature[B]{2D}{Two-Dimensional}%
\printnomenclature
\section{Introduction} \label{sec:introduction}
Differential Privacy (DP) is a technique used for publishing database queries that conceal confidential attributes. Some of its real-world applications are in the publication of the United States of America Census 2020 data (using disclosure avoidance system~\cite{US_Sensus_Pub}), Google's historical traffic statistics~\cite{GooglePolicy}, Microsoft's telemetry data~\cite{ding2017collecting}, LinkedIn user engagement information to third parties for advertisements~\cite{rogers2020linkedins}, etc.

DP hinges on a randomized mechanism wherein the data publisher, who owns the database, responds to analyst queries. The key principle is to generate similar distributions of query answers for data differing by a specific attribute, making it statistically challenging to discern whether data with that attribute were involved in the query computations.

In the existing literature, randomizing query answers commonly involves adding noise with a parametric distribution featuring one or two degrees of freedom. {In contrast, this paper proposes a novel approach: optimizing all the parameters of the probability mass function (PMF) for queries with finite discrete answers. This ensures that the randomized query outcome meets DP constraints while minimizing expected distortion for a given database and discrete finite set of answers. Notably, existing literature does not undertake the optimization of noise to minimize distortion under an arbitrary $(\epsilon, \delta)$ constraint, as it typically limits the mechanism to a single parameter. Our approach stands out because it optimizes the entire distribution.}

Our formulation, applicable to both discrete numerical data (for which perhaps mean squared error (MSE) is the best error metric) and categorical data (where error rate (ER) is preferable), finds resonance in real-world scenarios. {For instance, queries related to discrete numerical data include: 
(a) The number of households in a census tract with at least one college-educated member (0 through $n$, where $n$ is the total number of households in that census tract);
(b) what is the most popular promotion on a website;
or how many users in a certain set accepted a sales promotion on a website;
(c) The hour of peak electricity usage in a neighborhood (00 through 23), etc.
Similarly, queries related to categorical data are: 
(a) Type of consumer? (Subscriber or Trial user, Residential or Commercial, etc.);
(b) What month of the year? (January: 1 through December: 12);
(c) What gender is a person? (Male: 1, Female: 2, Other: 3), same idea for ethnicity, blood type, etc.
}
{\color{black} For discrete numerical data, just perform the modulo addition of random noise of size  $n+1$, and for categorical data, we assign numerical values to categories and perform the modulo addition of random noise on them.}

Before outlining our contributions, we review the relevant literature. 

\subsection{Literature review}
In the literature, several papers studied the additive noise mechanisms for discrete query outputs~\cite{Add_Geng, SORIACOMAS2013200, Approx_Geng}. 
For discrete queries with infinite support, the additive noise mechanism for $\epsilon-$differential privacy that minimizes any convex function of the query error was found in~\cite{Add_Geng}; the optimum PMF is shown to have a specific decreasing staircase trend. The problem of finding the optimal data-independent noise mechanism for $\epsilon-$differential privacy is also addressed  in~\cite{SORIACOMAS2013200}. Even though the authors focus on continuous query outputs, they claim one can easily extend the method to discrete queries. Neither paper~\cite{Add_Geng},\cite{SORIACOMAS2013200} explored the optimization of the $(\epsilon,\delta)-$differential privacy trade-off for $\delta>0$. For integer query outputs, the optimal noise mechanism design for $(\epsilon,\delta)-$differential privacy is the subject of~\cite{Approx_Geng}.  Another approach to integer count queries is carried out in~\cite{Ghosh_Universal}, where a double-sided geometric distribution noise mechanism is used. A recent study on the count query DP problem is found in~\cite{8716576}, in which the authors use a set of constrained mechanisms that achieve favorable DP properties.  The related problem of publishing the number of users with sensitive attributes from a database is addressed in~\cite{sadeghi2020differentially}.  In their proposed DP mechanism, they add an integer-valued noise before publishing it to protect the privacy of individuals. 
Though the randomized query response, produced by the proposed mechanism in~\cite{sadeghi2020differentially},  lies in the actual query support range, the additive noise PMF used depends on the query output, which requires storing several PMFs. In the context of discrete queries, an additive discrete Gaussian noise-based mechanism is proposed in~\cite{canonne2021discrete}. They show that the addition of discrete Gaussian noise provides the same privacy and accuracy guarantees as the addition of continuous Gaussian noise. Another recent study focuses on the mechanisms of discrete random noise addition~\cite{qin2022differential}. In this study, the basic DP conditions and properties of general discrete random mechanisms are investigated. In~\cite{Ravi2022Diff} a randomized mechanism for the labels obtained from $K$-means clustering is provided using the modulo addition-based mechanism.

In the literature, a joint DP mechanism is proposed for key-value data in~\cite{ye2019privkv}, where key uses categorical data and value uses numerical data. Two potential applications have been identified: video ad performance analysis and mobile app activity analysis. The key in the former is the ad identifier, and the value is the time a user has watched this video ad, whereas the key in the latter is the app identifier, and the value is the time or frequency this app appears in the foreground. In another work~\cite{wang2019collecting}, local DP mechanisms for multidimensional data
that contain both numeric and categorical attributes are proposed.

\subsection{Paper contributions}
{This paper critically revisits the design of DP random mechanisms, specifically focusing on ensuring $(\epsilon,\delta)-$DP for queries with $n+1$ possible answers, each mapped onto the integers $0$ through $n$. The mechanism we study involves the modulo $n+1$ addition of noise. The significant and novel contributions of this paper can be summarized as follows:}
\begin{itemize}
    \item {\textbf{Optimal Noise PMF:} In Section \ref{Sec:Opt_numerical}, we demonstrate that the additive noise PMF minimizing a linear error metric under a given $(\epsilon, \delta)$ budget can be obtained as the solution of a Mixed Integer Linear Program (MILP). Notably, for the case when $\delta=0$, the optimum PMF can be found using a Linear Program (LP), as established in previous literature~\cite{Stair_Geng}, which is a special case of our general formulation.} 
    \item {\textbf{Explicit PMF expressions for minimum error:} Sections \ref{Sec:Opt_Noise_SD_Nbd} and \ref{Sec:Opt_Noise_BD_Nbd} delve into two special cases, providing explicit expressions for the optimum PMF that minimizes error for specific $(\epsilon,\delta)$ pairs. This analysis extends and subsumes prior work~\cite{Stair_Geng}.}
    \item {\textbf{Structure of Optimum PMF and Error Rate:} We unveil the structural characteristics of the optimum PMF and error rate functions. Specifically, the derived error rate function exhibits a piece-wise linear nature. Our findings reveal that the optimum $(\epsilon,\delta)$ trade-off curve, for a given error rate, experiences an exponential decrease as $\delta$ increases. Moreover, a discrete set of discontinuities in the curve precludes any change in the exponential rate of decay as $\epsilon$ increases.}
    \item {\textbf{Numerical Validation and Comparative Analysis:} The contributions outlined above are corroborated by a thorough numerical analysis presented in Section \ref{Sec:Simulations}. Our simulations include comprehensive comparisons with prior methods, demonstrating the superiority of our proposed approach.}
\end{itemize}

\subsection{Notation} \label{Section:Notation}
{In this section, we present a summary of general notations used throughout the paper. A comprehensive list of notations and abbreviations can be found at the top of the paper in the Nomenclature section.}

Let $\mathbb{N}, \mathbb{N}^+, \mathbb{Z}, \mathbb{R}$ denote the sets of natural numbers including zero, natural numbers excluding zero, integers, and real numbers, respectively.
The set integers $\{0,1, \ldots, n\}$, $n \in  \mathbb{N}$, is referred to as $[n]$, $[n]_{+}$ is, instead, $\{1, \ldots, n\}$. 
The symbols $\lfloor k \rfloor$ and $\lceil k \rceil$ denote the floor and ceiling functions of $k$, respectively. The cardinality of set $\mathcal{A}$ is denoted by $|\mathcal{A}|$, and the set difference of sets $\mathcal{A}$ and $\mathcal{B}$ is denoted as $\mathcal{A} \setminus \mathcal{B}$.  

The query function applied to data $X$ from a database, denoted by $\mathcal X$, is represented as $\mathbbm{Q}(X)$, and $\mathcal{Q}$ denotes a discrete finite set of query answers. In this paper, the query domain is discrete and finite, mapped onto the set $[n]$ of size $n+1$. The numerical outcome of the query is denoted by the variable $q\in [n]$, while $\tilde{q}$ represents the outcome after the randomized publication, with distribution $f(\qtilde|X)$. 

For vector queries with outcomes, $\bm q$, in a finite discrete domain, one can map the result onto the set $[n]$, and hence, the optimization we propose applies. We use $f(\eta)$ to represent the noise distribution $f_\eta(\eta)$ whenever possible without confusing the reader.

\section{Preliminaries}
{
\begin{defn}[$(\epsilon,\delta)$-Differential Privacy (DP)~\cite{Calib_Dwork}]
A randomized mechanism $\qtilde:\mathcal{Q} \rightarrow \mathcal{Q}$ is $(\epsilon,\delta)$-differentially private if for all datasets $X$ and $X'$ differ by a unique data record, given any arbitrary event $\mathcal{S} \subseteq \mathcal{Q}$ pertaining to the outcome of the query, the randomized mechanism satisfies the following inequality
\begin{equation}\label{eq:DP}
      Pr(\qtilde(X) \in \mathcal{S}) \leq e^\epsilon Pr(\qtilde(X') \in \mathcal{S}) + \delta,
\end{equation}
where $Pr(\mathcal{A})$ denotes the probability of the event $\mathcal{A}$ and the PMF used to calculate the events probability is $f(\tilde{q}|X)$.
\end{defn}
}
Conventionally, given the random published answer $\qtilde$ in the differential privacy literature, the {\it privacy loss} function name is a synonym for the log-likelihood ratio:  
\begin{equation} \label{Eq:PrivacyLoss}
L_{\XX}(\qtilde)\triangleq \ln \frac{f_{\qtilde}(\qtilde|X)}{f_{\qtilde}(\qtilde|X')},
\end{equation}
where $X \in {\mathcal X}$ is the set of data used to compute the query and  $X'$ is an alternative set with a unique attribute or data point that is different. For each $X$ we denote by $ {\mathcal X}^{(1)}_{X}$
the neighborhood of set of $X$ which contains all data sets $X'\in {\mathcal X}$ that differ from $X$ by a predefined sensitive attribute we want to conceal. 
Note that, if the event $L_{\XX}(\qtilde)<0$ under the experiment with distribution $f(\qtilde|X)$ then, in classical statistics,  the observer of the outcomes $\tilde{q}$ will choose erroneously the alternative hypothesis that $X'$ was queried (where the emission probability is $f(\qtilde|X')$) rather than $X$. By looking at the tail of the distribution for $L_{\XX}(\qtilde)>0$ under the distribution $f(\qtilde|X)$, one can gain insights into how frequently the mechanism allows to differentiate $X$ from $X'$ with great confidence, leaking private information to the observer. 

We now introduce the definition of $(\epsilon, \delta)-$ probabilistic differential privacy (PDP) we consider in this paper, which applies to any random quantity $\tilde{q}$ for any given $X$:
\begin{defn}[\bf $(\epsilon, \delta)$-{Probabilistic} DP~\cite{Privacy_Machanavajjhala}]\label{defn:e-d-privacy} Consider random data that can come from a set of emission probabilities $\tilde{q} \sim f(\tilde{q}|X)$ that change depending on $X\in {\mathcal X}$. The data $\tilde{q}$ are  $(\epsilon, \delta)$- probabilistic differentially private $\forall X \in {\mathcal X}$ and  $X^{'}\in {\mathcal X}^{(1)}_{X}$, iff: 
\begin{align}
  \delta&\geq \delta^{\epsilon}_{\qtilde} \triangleq \sup_{X\in {\mathcal X}} \sup_{X'\in {\mathcal X}_{X}^{(1)}} Pr
    \left(L_{\XX}
(\tilde{q})>\epsilon\right)\label{eq:epsilon_delta_relation},
\end{align}
and the PMF used to calculate the probability is $f(\tilde{q}|X)$.
\end{defn}

The following theorem guarantees that $(\epsilon,\delta)$-PDP is a strictly stronger condition than $(\epsilon,\delta)$-DP. 
\begin{thm}[PDP implies DP~\cite{mcclure2015relaxations}]\label{thm:PDP-DP}
If a randomized mechanism is $(\epsilon,\delta)$-PDP, then it is also $(\epsilon,\delta)$-DP, i.e.,
\[
    (\epsilon,\delta)-\text{PDP} \! \Rightarrow \! (\epsilon,\delta)-\text{DP}, \text{but } (\epsilon,\delta)-\text{DP} \! \nRightarrow \! (\epsilon,\delta)-\text{PDP}.
\]
\end{thm}
The proof of Theorem~\ref{thm:PDP-DP} is shown in~\cite{mcclure2015relaxations,triastcyn2019improved}. This motivates us to use $(\epsilon,\delta)-$PDP throughout this paper. 
{\color{black}The $(\epsilon,\delta)-$PDP has applications in the contexts of location recommendations~\cite{8306835}, privacy protection from sampling and perturbations~\cite{sholmo2012privacy}, to create a realistic framework for statistical agencies to distribute their goods~\cite{Privacy_Machanavajjhala}, etc.}

In our setup $\mathbbm{Q}(X)$ is a scalar value in ${\mathcal Q}\equiv [n]$.
The query response  $\tilde{q}$ is obtained by adding a discrete noise $\eta$, whose distribution is denoted by $f(\eta)$, i.e.:
{
\begin{align} \label{eq:AdditiveNoiseDefn}
    \tilde{q}=\mathbbm{Q}(X)+{\eta}~~\Rightarrow~~f_{\tilde{q}}(\tilde{q}|X)=f_{\eta}(\tilde{q}-\mathbbm{Q}(X)).
\end{align}
}
The PMF associated with the privacy loss function, called privacy leakage probability, for the additive noise mechanism, can be derived from ~\eqref{Eq:PrivacyLoss},~\eqref{eq:epsilon_delta_relation} and~\eqref{eq:AdditiveNoiseDefn}:
\begin{align} \label{eq:noise-only-classic-DP}
    Pr(L_{\XX}(\tilde{q})>\epsilon)=Pr\left(
    \ln\frac{f_{\eta}(\tilde{q}-\mathbbm{Q}(X))}{f_{\eta}(\tilde{q}-\mathbbm{Q}(X'))}>\epsilon
    \right).
\end{align}
For the discrete query case, we denote the  ``distance one set'' of $X \in {\mathcal X}$ as   ${\mathcal X}_{X}^{(1)} \subset {\mathcal X} \setminus X$ and let:
 \begin{align}
     \mu_{\XX}\triangleq \mathbbm{Q}(X)-\mathbbm{Q}(X'), ~~~~~ \forall X \in {\mathcal X},\forall X'\in {\mathcal X}_{X}^{(1)}. \label{eq:mu_xxDQ}
 \end{align} 
 where $X'$ differs from $X$ for one user data record or a sensitive user attribute.
Let us define the indicator function  $u_{\XX} (\eta)$, $\eta \in [n]$ such that $i^{th}$ entry is one if $L_{\XX}(\tilde{q}_i)>\epsilon$ and zero otherwise, i.e.:
\begin{align}
u_{\XX} (\eta) \triangleq & 
    \begin{cases}
    1, & ~~~~~f(\eta) > e^\epsilon f(\eta+ \mu_{\XX})\\
    0, &~~~~~~\mbox{otherwise}
    \end{cases}  
    \label{eq:u_XX_Defn}
\end{align}{
where we omitted the suffix $\eta$ and used $f(\eta)$ in lieu of $f_{\eta}(\eta)$.}
It is easy to verify that we have:
$$ Pr(L_{\XX}(\tilde{q})>\epsilon)=\sum_{\eta=0}^{n}
u_{\XX} (\eta) f(\eta).$$

Before describing our design framework in Section \ref{Sec:Modulo_Addition}, a few considerations on $(\epsilon,\delta)-PDP$ are in order. 

\subsection{{Post-processing}}\label{sec:PostProcess}
A randomized DP mechanism maps a query output onto a  distribution designed to meet Definition \eqref{eq:DP} or \eqref{eq:epsilon_delta_relation}.  {\color{black} If the randomized query answer generation requires multiple steps, it is important to ensure that the $(\epsilon,\delta)-PDP$ ($(\epsilon,\delta)-DP$) are met after the very last step.} In fact, in \cite{meiser2018approx} it was pointed out that, unless $\delta=0$, in general $(\epsilon,\delta)-PDP$ with $\delta>0$ cannot be guaranteed after {\it post-processing}, {\color{black} where post-processing refers to processing steps that follow the noise addition prior to the release the query response}.  
The objections in \cite{meiser2018approx} are valid for mechanisms that include post-processing like {\color{black} the popular  ``truncation'' or ``clamping'' mechanisms that consists of first adding unbounded noise $\eta$ and then projecting (clamping) the sum $\mathbb{Q}(X)+\eta$ in the prescribed range to generate $\tilde q$}. In this case $(\epsilon, \delta)-DP$ are guaranteed before the post-processing step, but not after, unless $\delta=0$. 
The proposition below  provides guarantees for $(\epsilon,\delta)-{PDP}$. 
\begin{prop} \label{prop:PostProcessing} Let $\qtilde \in {\mathcal Q}$ be a randomized $(\epsilon,\delta)-\text{PDP}$ response for query $q \in {\mathcal Q}$. Let $g: {\mathcal Q} \rightarrow {\mathcal Q}$ be an arbitrary invertible mapping. Then $g \circ \qtilde = g(\qtilde)$ is also a $(\epsilon,\delta)-{PDP}$ response for any $\delta\geq 0$. Furthermore, if $\delta=0$, $\epsilon-{PDP}$ is preserved, irrespective of $g$. 
\end{prop}
\begin{proof}
The proof is in Appendix~\ref{Apdx:prop1Proof}.
\end{proof}
 Proposition~\ref{prop:PostProcessing}  clarifies the importance of designing randomized responses whose domain is consistent with the query output, since it does not require post-processing to generate answers in the right set. Clamping is not bijective and changes the masses of probability in a way that alters $\delta$ for a given $\epsilon$.

\section{Optimal Additive Noise}\label{Sec:Modulo_Addition}
For queries $q = \mathbbm{Q}(X) \in [n]$, a possible approach other than clamping is to assume that the noise addition is modulo $n+1$ with $\eta \in [n]$ so that the outcome $\tilde{q}=q + \eta ~(\text{mod } n+1)$, is always in the appropriate range. 
In this paper, we seek to obtain the optimum noise distribution $f(\eta)$ for such a mechanism. 
 Since $(\epsilon, \delta)-PDP$ implies $(\epsilon, \delta)-DP$ and hence, it is a stronger notion, we use the definition of $(\epsilon, \delta)-PDP$ throughout.   

Next, we omit the $(\text{mod } n+1)$  to streamline the notation, with the understanding that, from now on, sums and differences of query outcomes and noise values are always modulo $n+1$.  Observe that adding uniform noise would lead necessarily to a scalar query $\tilde q$ being uniform and thus, high privacy (i.e. $\delta=0$ for any $\epsilon>0$) but poor accuracy since $1-1/(n+1)$. This motivates the search for an optimal solution. 
Using~\eqref{eq:noise-only-classic-DP} and~\eqref{eq:mu_xxDQ}:
\begin{align}
    Pr(L_{\XX}(\tilde{q})>\epsilon)&=Pr\left(\ln\frac{f(\eta)}{f(\eta + \mu_{\XX})}>\epsilon
    \right).\label{eq:noise-only-modulo}
\end{align} 
The reasons for using the modulo addition of noise are:
\begin{itemize}
    \item  The randomized answers fall within the range expected for the query, which allows us to leverage  Proposition~\ref{prop:PostProcessing}.
    \item The mechanisms require defining a single distribution rather than distinct distributions for all possible $X\in {\mathcal X}$.
    \item In the optimization, any pair with the same modulo difference results in a single $(\epsilon, \delta)-PDP$ constraint, simplifying the search for the optimum distribution. 
    \item The simplifications allow us to derive the optimum distribution in closed form for specific use cases. 
\end{itemize}

From~\eqref{eq:noise-only-modulo} it is evident that the $(\epsilon,\delta)$ privacy curve is entirely defined by the noise distribution and its change due to a shift in the mean. 
As a result, the probability mass $f(q+\eta)$ is obtained as a circular shift of the PMF $f(\eta)$; therefore \eqref{eq:noise-only-modulo} can be used with the denominator $f(\eta+\mu_{\XX})$ also representing a circular shift of $f(\eta)$. This result motivates us to define the neighborhood sets, using only  $\mu_{\XX}$, in Section~\ref{sec:analytics}. 

\subsection{Numerical Optimization}\label{Sec:Opt_numerical}

In this section, we show that the problem of finding an optimal additive noise mechanism for a given pair $(\epsilon,\delta)$ and expected distortion cost can be cast into a MILP formulation, i.e. an optimization problem with linear cost, linear equality and inequality constraints, and real as well as integer variables. While MILPs are non-convex, several stable solvers have convergence guarantees. 
{
Our MILP formulation (in~\eqref{eq:MILP_obj}-\eqref{eq:eta_range}) finds optimum noise distributions minimizing a specified expected distortion cost:
\begin{equation}
    \mathbb{E}[\rho_{\eta}]=
    \sum_{\eta=0}^n \rho_{\eta}
    f(\eta),
\end{equation}
{where $\rho_{\eta}$ denotes the distortion caused by the noise value  $\eta \in [n]$.}
There are two typical metrics:
\begin{defn}[\bf Error Rate]\label{def:discrete-accuracy} 
For $ \tilde{q}=\mathbbm{Q}(X)+ \eta$, the error rate metric, denoted by $\rho^{\text{ER}}$, is the expected value of the function $\rho_0=0$ and $\rho_{\eta}=1, \eta>0$. Thus:
\begin{equation}
\rho^{\text{ER}} \triangleq 1 - f(0).\label{eq:error-rate}
\end{equation}
\end{defn}
 \begin{defn}[\bf Mean Squared Error]\label{def:contaccuracy}
For $\tilde{q}=\mathbbm{Q}(X) + \eta$, the MSE corresponds to $\rho_{\eta}=\eta^2$:
\begin{align}\label{eq:ms-sensitivity}
    \rho^{\text{MSE}} \triangleq \mathbb{E}[|\tilde{q}-\mathbbm{Q}(X)  |^{2}] = \sum_{\eta=1}^{n} \eta^{2} f(\eta).
\end{align}
\end{defn}
{
\begin{rem}
\label{Rem:MILP2}
In the case of an `ordered' query domain, the MSE metric can be preferable over an error rate metric. That is why we are focusing on the minimization of any linear cost in the MILP formulation. However, the numerical results shown in Section~\ref{Sec:Simulations} indicate that even when we target the error rate metric $\rho^{ER} = 1 - f(0)$, the optimal solution tends to have values that diminish as they move away from the actual query. 
\end{rem}
}
Because our analytical results in Section \ref{sec:analytics} consider the error rate metric, whenever $\rho$ is mentioned without specification, this implies $\rho^{ER}$ is being discussed.  
Having established the cost, the constraints (see~\eqref{eq:Sum_f_eta}-\eqref{eq:integrality}) are derived as follows.
From the databases, we calculate the set $\{\mu_{\XX}\}$ using~\eqref{eq:mu_xxDQ}, which are the only database parameters needed in the formulation. 
} 
As we know, the sum of probability masses is 1 (see~\eqref{eq:Sum_f_eta}).
The indicator function $ u_{\XX}(\eta)$, defined in~\eqref{eq:u_XX_Defn}, can be mapped on the integrality constraint, $u_{\XX} (\eta) \in \{0,1\}, \forall X \in {\mathcal X},\forall X'\in {\mathcal X}_{X}^{(1)}$ (see \ref{eq:integrality}) and on two linear inequality constraints, $f(\eta) - e^\epsilon f(\eta+\mu_{\XX}) - u_{\XX} (\eta) \leq 0$ and $e^\epsilon f(\eta+\mu_{\XX}) - f(\eta) + e^{\epsilon}u_{\XX} (\eta) \leq e^{\epsilon}$ (see~\eqref{eq:f_eta_constraint1} and~\eqref{eq:eta_range}).\footnote{The indicator function $u_{\XX} (\eta)$ (in short $u$) can take two values (see~\eqref{eq:u_XX_Defn}). In the first case, when $u=1$, the equations~\eqref{eq:f_eta_constraint1} and~\eqref{eq:eta_range} become $f(\eta) - e^\epsilon f(\eta+\mu_{\XX}) \leq 1$ and $e^\epsilon f(\eta+\mu_{\XX}) - f(\eta) \leq 0$. In~\eqref{eq:f_eta_constraint1}, since $\epsilon > 0$ and $f(\eta) \in [0,1]$, the maximum difference between $f(\eta)$ and $e^\epsilon f(\eta+\mu_{\XX})$ is 1 which happens when $f(\eta) = 1$ and $f(\eta+\mu_{\XX}) = 0$, and for the rest of the cases difference is less than 1. The equation~\eqref{eq:eta_range} is true due to the definition of $u$ (see~\eqref{eq:u_XX_Defn}). Next, when $u=0$, the equations~\eqref{eq:f_eta_constraint1} and~\eqref{eq:eta_range} become $f(\eta) - e^\epsilon f(\eta+\mu_{\XX}) \leq 0$ and $e^\epsilon f(\eta+\mu_{\XX}) - f(\eta) \leq e^\epsilon$. In this case,~\eqref{eq:f_eta_constraint1} is true due to the definition of $u$ (see~\eqref{eq:u_XX_Defn}) and the other equation can be modified by multiplying $e^{-\epsilon}$ both sides to $f(\eta+\mu_{\XX}) - e^{-\epsilon} f(\eta) \leq 1$, which is true due to the similar arguments made earlier for the case $u=1$.}
We can rewrite~\eqref{eq:noise-only-modulo} as  $Pr(L_{\XX}(\tilde{q})>\epsilon) = Pr(f(\eta) > e^\epsilon f(\eta+ \mu_{\XX}))= \sum_{\eta=0}^n u_{\XX} (\eta) f(\eta) \le \delta$, $\forall X \in {\mathcal X},\forall X'\in {\mathcal X}_{X}^{(1)}$. Since it is a bilinear constraint, {and for a MILP formulation we need the constraints to be linear}, we introduce the auxiliary variables $y_{\eta}, \eta \in  [n]$:
\begin{align}
\label{eq:y_Defn}
    y_{\eta} &\triangleq u_{\XX} (\eta)  f(\eta), ~~~~\eta \in  [n].
\end{align}
so that we can use 
$ \sum_{\eta=0}^n y_{\eta} \le \delta$ to constrain $\delta$ instead; to enforce $y_{\eta}=f(\eta)$ for $u_{\XX}(\eta)=1$, and $y_{\eta}=0$ for $u_{\XX}(\eta)=0$, the trick is to use, respectively, the following linear constraints: $ u_{\XX} (\eta) + f(\eta) - y_{\eta} \le 1$, $u_{\XX} (\eta) -f(\eta) +  y_{\eta} \le 1$,\footnote{\label{fn:y_f_relation}This inequality ensures that $y_{\eta} =   f(\eta)$ if $u_{\XX} (\eta) = 1$.} and $y_{\eta} - u_{\XX} (\eta) \le 0$.\footnote{This inequality ensures that $y_{\eta} =  0$ when $u_{\XX} (\eta) = 0$.} 
This completes the explanation of the optimization constraints shown in~\eqref{eq:y_eta_condition}-\eqref{eq:Aux2}.\footnote{Note that other constraints can be added, for instance, that of having a zero mean distribution, or forcing some of the PMF values to be identical, etc.} 
Let
\begin{equation*}
{\cal M}\!\triangleq\!\{\mu_{\XX}| \mu_{\XX}=\mathbb{Q}(X)-\mathbb{Q}(X'),\forall X\!\!\in\!\! {\mathcal X},\forall X'\!\!\in\!\! {\mathcal X}_{X}^{(1)}\}    
\end{equation*}
The form of the MILP is:
\begin{subequations}\label{eq:MILP}
\begin{align}
&{~~~~~~\min_{f(\eta),u_{\XX}(\eta), y_{\eta}} ~~\sum_{\eta=0}^{n}\rho_{\eta}f(\eta)} \label{eq:MILP_obj}\\
&~\text{s.t.}~\sum_{\eta=0}^n f(\eta)= 1, \label{eq:Sum_f_eta} \\
&~~~~~~\sum_{\eta=0}^n y_{\eta} \le \delta,  \label{eq:y_eta_condition} ~~~~~~y_{\eta} - u_{\XX} (\eta) \le 0, \\
&~~~~~~ u_{\XX} (\eta) + f(\eta) - y_{\eta} \le 1, \label{eq:Aux1}\\
&~~~~~u_{\XX} (\eta) -f(\eta) +  y_{\eta} \le 1, \label{eq:Aux2} \\
&~~~~~~ f(\eta) - e^\epsilon f(\eta+\mu_{\XX}) - u_{\XX} (\eta) \leq 0, \label{eq:f_eta_constraint1}\\
&~~~~~~ e^\epsilon f(\eta+\mu_{\XX}) - f(\eta) + e^{\epsilon}u_{\XX} (\eta) \leq e^{\epsilon},\label{eq:eta_range} \\
  &~~~~~~f(\eta),~y_{\eta} \in [0, 1], ~\eta \in  [n], \nonumber \\
  &~~~~~~u_{\XX} (\eta) \in \{0,1\}, \forall \mu_{\XX} \in \mathcal{M}. \label{eq:integrality} 
\end{align}
\end{subequations}
The MILP has $3(n+1)$ variables out of which $f(\eta)$ and  $y_{\eta}$, $\eta \in [n]$ are real numbers ($2(n+1)$ in total) in $[0,1]$ and $u_{\XX} (\eta)$ (also $(n+1)$ in number) are integers $\{0,1\}$. {The computational complexity of the MILP solution is determined by the algorithm used, which can be branch and bound, cutting plane, or branch and cut. A detailed study on the complexity of the MILP solution is found in~\cite{jiang2021complexity}.}   
When $\delta = 0$, it reduces  to the following LP:
\begin{subequations}
\begin{align}
&~~~~~~\min_{f(\eta),u_{\XX}(\eta), y_{\eta}} \sum_{\eta=0}^{n}\rho_{\eta}f(\eta)  \label{Eq:Opt_prob}\\
&~~\text{s.t.}~\sum_{\eta=0}^n f(\eta)= 1, \label{eq:prob_distr_condition}~~~
\\
&~~~~~~ f(\eta) - e^\epsilon f(\eta+\mu_{\XX})\leq  0,~~\eta \in [n],  \label{Eq:f0_constraint} \\
&~~~~~~f(\eta) \in [0,1], ~\forall \mu_{\XX} \in \mathcal{M}
\end{align}
\end{subequations}
A possible useful variant of the optimization in \eqref{eq:MILP} that will be explored in our numerical results is to minimize $\delta$ instead, under a distortion constraint $\overline{\rho}$, i.e.:
\begin{subequations}\label{eq:MILP-delta}
\begin{align}
&~~~~~~\min_{f(\eta),u_{\XX}(\eta), y_{\eta}} \delta \label{eq:MILP_obj2}\\
&~~\text{s.t.}~\sum_{\eta=0}^{n}\rho_{\eta}f(\eta)\le \overline{\rho},  \label{eq:rho_condition} \\
  &~~~~~~\eqref{eq:Sum_f_eta}-\dots-\eqref{eq:integrality} 
\end{align}
\end{subequations} 
In the next sections, we derive analytical solutions of \eqref{eq:MILP} that minimize the error probability
$\rho^{ER} = 1 - f(0)$
for some special database structures and corroborate the results in Section~\ref{Sec:Simulations} comparing the formulas with the MILP solutions obtained using Gurobi~\cite{gurobi} as a solver.

\subsection{Analytical Solutions}\label{sec:analytics}
{In this section, we analytically study the solution that minimizes the error rate $\mathbb{E}[\rho_{\eta}]=1-f(0)$}.
To give closed-form solutions for the optimum PMF, we focus on the following instances of possible $\mu_{\XX}$:
\begin{defn}[\bf Single Distance (SD)] In this setting $\forall X \in {\mathcal X},\forall X'\in {\mathcal X}_{X}^{(1)}$, the difference $\mu_{\XX}$ is constant, i.e. $\mu_{\XX} = \hat{\mu}$. Note that $\hat{\mu} \le n$. \label{def:FixedDiffNbd}
\end{defn}
\begin{defn}[\bf Bounded Difference (BD)] In this setting $\forall X \in {\mathcal X},\forall X'\in {\mathcal X}_{X}^{(1)}$ $\mu_{\XX}$ take values in the  set $[\bar{\mu}]_+$, $\bar{\mu} \le n$, at least once.\label{def:FixedRangeNbd}
\end{defn}
The most general case is the following:
\begin{defn}[\bf Arbitrary]  In this case $ \mu_{\XX}$ can take values from any subset of  $[n]$,$~\forall X \in {\mathcal X},\forall X'\in {\mathcal X}_{X}^{(1)}$.\label{def:GeneralNbd}
\end{defn}
The next lemma clarifies that the optimum PMF for the BD case for a given $(\epsilon,\delta)$ is useful to attain the same guarantees for the case of an arbitrary neighborhood. 
\begin{lem}
Suppose $\bar{\mu}=\sup_{\forall X \in {\mathcal X},\forall X'\in {\mathcal X}_{X}^{(1)}}\mu_{\XX}$. Then any noise $PMF$ that provides $(\epsilon,\delta)$ privacy for the BD neighborhood with parameter $\bar{\mu}$ will give the same guarantees in terms of  $(\epsilon,\delta)$ and $\rho$ for the case of arbitrary neighborhoods. However, lower distortion is achievable by solving the MILP in \eqref{eq:MILP}.  
\end{lem}
\begin{proof}
The proof is simple: the set of constraints that need to be met to satisfy  $(\epsilon,\delta)$ privacy for the arbitrary case is a subset of the BD neighborhood case with $\bar{\mu}$ as a parameter. This also means, however, that the minimum $
\rho^\star$ for the latter case is sub-optimum. 
\end{proof}

Note that the SD neighborhood setting is a simple case while the previous lemma indicates that the BD case is more useful in general. 

Next, we find an explicit solution for the optimum noise PMF $f^\star(\eta), \eta \in [n]$ for the SD and BD neighborhood cases. In Section~\ref{Sec:Opt_Noise_Vector}, we discuss the case of discrete vector queries whose entries are independently subjected to the mechanism vs. the optimal solution. 

\subsubsection{PMF for Single Distance  neighborhood} \label{Sec:Opt_Noise_SD_Nbd}
The natural way to express the optimum PMF for an SD neighborhood setting is by specifying the probability masses sorted in decreasing order of $f(k)$, $\forall k \in [n]$. They are referred to as $\{f_{(k)}| 1 \ge f_{(0)} \ge f_{(1)} \ge \ldots \ge f_{(n)} \ge 0\}$. 
\begin{lem}\label{Lemma:ordered_set} Considering the case in Definition \ref{def:FixedDiffNbd} where  $\mu_{\XX} = \hat{\mu}$ is constant $~\forall X \in {\mathcal X},\forall X'\in {\mathcal X}_{X}^{(1)}$, the noise PMF minimizing the error rate is such that
$f^\star_{(h)} = f^\star(h\hat{\mu})$, $\forall h \in [n]$ and the inequality in~\eqref{Eq:f0_constraint}, can be written in terms of the sorted PMF as follows:
\begin{align}
    \label{Eq:f_eta_ordered}
    f^\star_{(h)} - e^\epsilon f^\star_{(h+1)}\leq  0,~~\forall h \in [n].
\end{align}
\end{lem}
\begin{proof}
The proof is in Appendix~\ref{Apdx:Lem1Proof}.
\end{proof}
To start, let us consider the case $\delta = 0$:
\begin{lem}\label{Lem:SD_delta0}
 For the SD neighborhood and  $\delta = 0$, the optimal noise PMF for the modulo addition mechanism is:\\
\textbf{Case 1}: If $(\hat{\mu}, (n+1))$ are relatively prime.
\begin{subequations}
\begin{align} \label{eq:Noise_IntegerCase1}
      f_{(k)}^\star &=  e^{-k \epsilon} f_{(0)}^\star, ~~~~k \in [n]_+,\\
            f_{(0)}^\star& = \frac{1 - e^{-\epsilon}}{1 - e^{-(n+1)\epsilon}}\equiv f^\star(0).\label{eq:SD_f*(0)}
\end{align}
\end{subequations}
\textbf{Case 2}: If $(\hat{\mu}, (n+1))$ are not relatively prime.
\begin{subequations}
\begin{align} \label{eq:Noise_IntegerCase2}
      f_{(k)}^\star ~&=  e^{-k \epsilon} f_{(0)}^\star, ~~k \in [N_{\hat{\mu}}-1]_+\\
      f^\star_{(k)} &= 0, \quad \quad ~~~~~k \in [n] \setminus [i\hat{\mu}], ~\forall i \in [N_{\hat{\mu}}-1]\\
      f_{(0)}^\star& = \frac{1 - e^{-\epsilon}}{1 - e^{-(N_{\hat{\mu}})\epsilon}}\equiv f^\star(0).\label{eq:SD_f*(0)_NotPrime}
\end{align}
\end{subequations}
where $N_{\hat{\mu}} = \frac{(n+1)}{\gcd((n+1),\hat{\mu})}$
and $\rho^{\star} = 1-f^\star(0)$.
\end{lem}  
\begin{proof}
The proof is in Appendix~\ref{Apdx:Cor1Proof}.
\end{proof}
Lemma~\ref{Lem:SD_delta0} is verified numerically in Section~\ref{Sec:Simulations} in both Case 1 and 2 (see Fig.~\ref{Fig:SD_PMF_plots}). 
We illustrate the two cases in Fig.~\ref{Fig:Periodic_Explain}. 

\begin{figure}[h!] 
    \centering
    \subfloat[][
    Case 1: $\hat{\mu} = 3, ~n =7$, here (3, 8) are relatively prime.]{\includegraphics[width=0.445\linewidth]{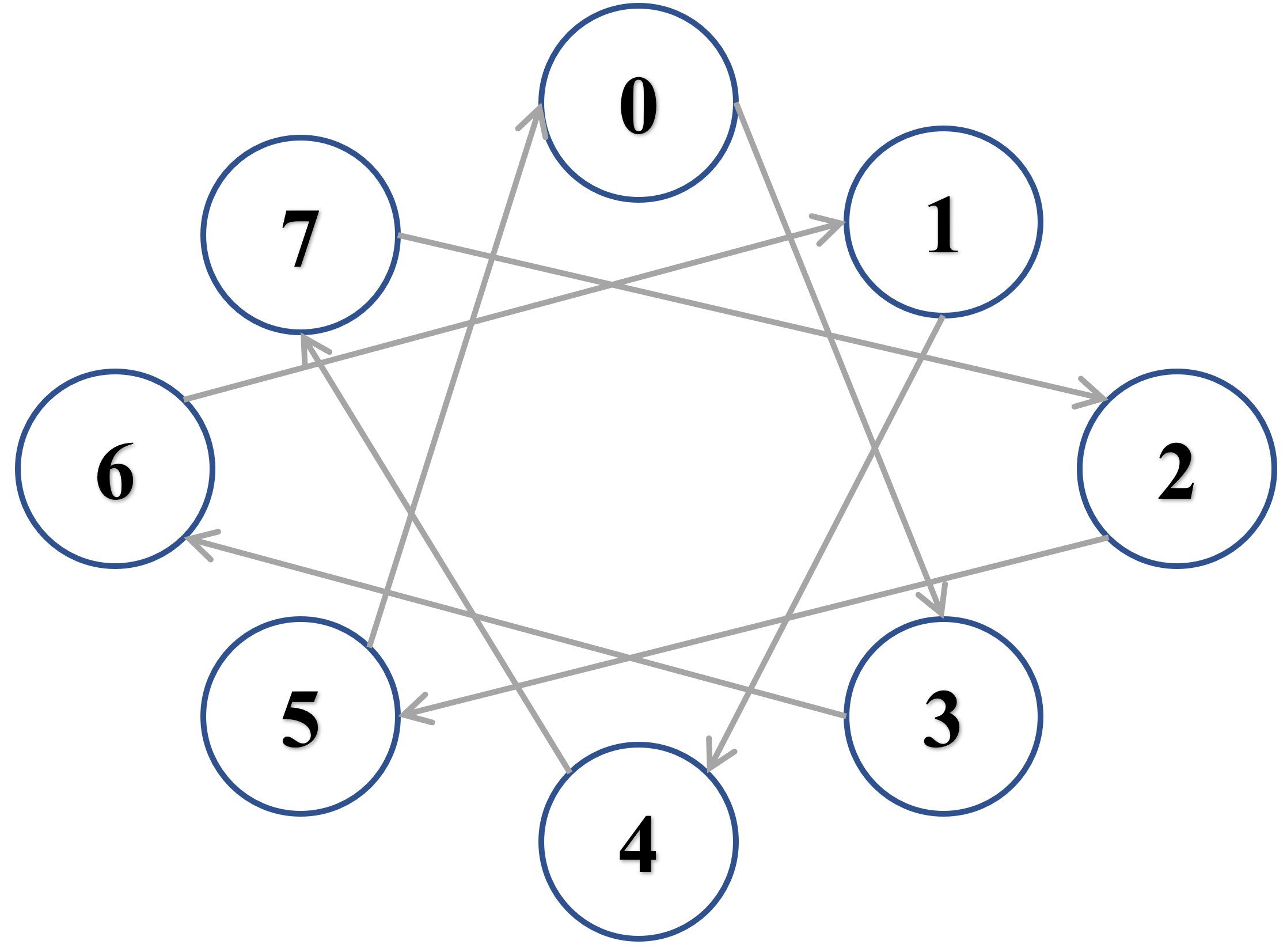}\label{fig:Case1} }~
    \subfloat[][
    Case 2: $\hat{\mu} = 2, ~n =7$, here (2, 8) are not relatively prime. 
    ]{\includegraphics[width=0.435\linewidth]{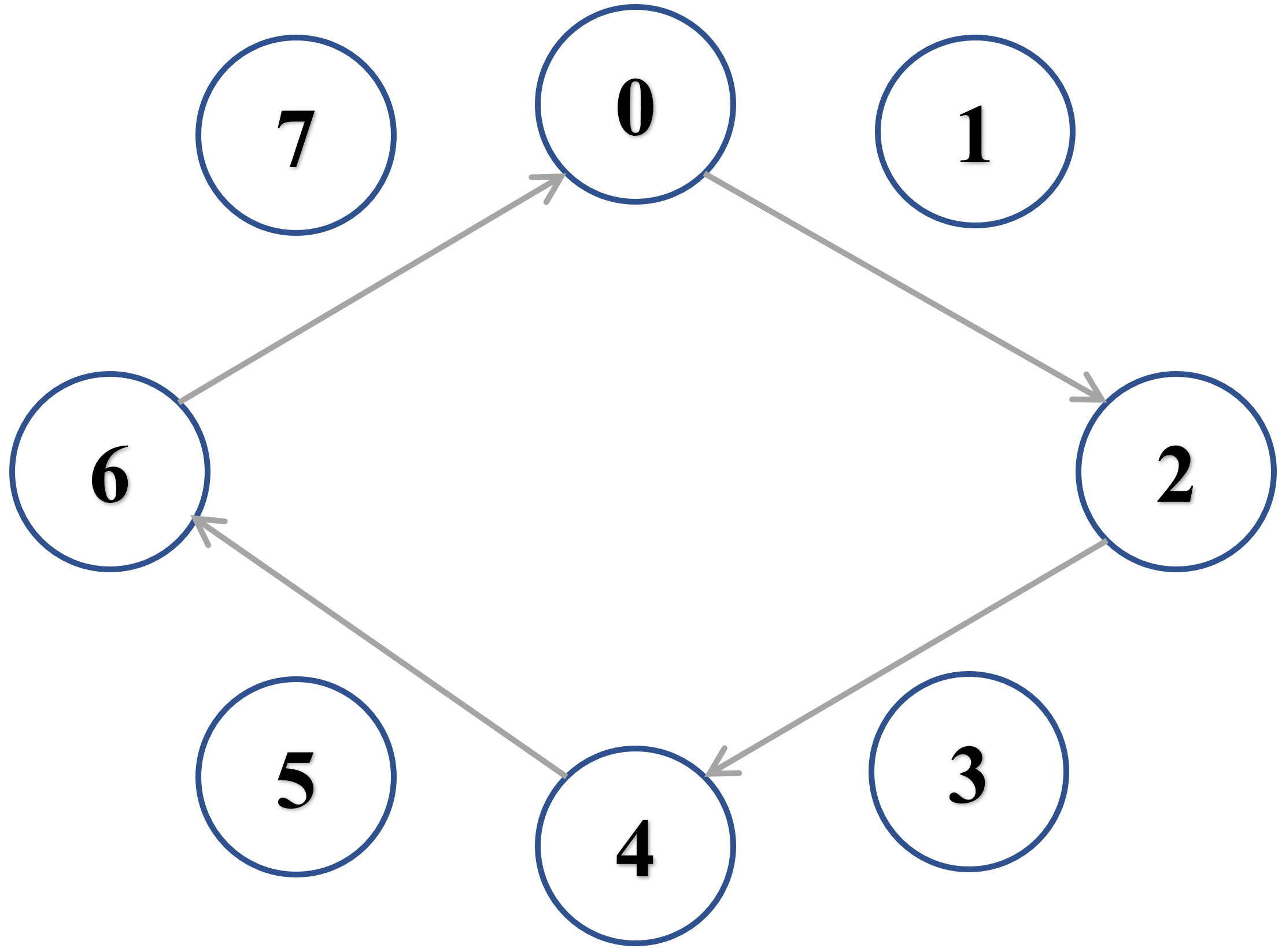}\label{fig:Case2} }
    \caption{In these examples $f^\star{(\hat{\mu})}$ is single distance, $\hat{\mu}$, away from $f^\star{(0)}$ hence it is assigned $e^{-\epsilon}f^\star{(0)}$, next  $f^\star{(2\hat{\mu})}$ is assigned $e^{-\epsilon}f^\star{(\hat{\mu})}$ since it is $\hat{\mu}$ away from $f^\star{(\hat{\mu})}$ and so on and so forth. So the order of assignment of values for plot (a) example is: $0,3,6,1,4,7,2,5$ and the order of assignment of values for plot (b) example is: $0, 2, 4, 6$. Since the values at $1, 3, 5, 7$ are not $\hat{\mu}$ away from $f^\star{(2k\hat{\mu})}, k \in [N_{\hat{\mu}}-1=3]$ they are assigned $0$ value to have a higher $f^\star{(0)}$.}
    \label{Fig:Periodic_Explain}
\end{figure}
From \eqref{eq:SD_f*(0)} we can observe that for case 1, the error rate $\rho^{\star}$ depends only on $\epsilon$ and $n$, and for case 2 (see \eqref{eq:SD_f*(0)_NotPrime}), $\rho^{\star}$ depends on both $\epsilon$ and  $N_{\hat{\mu}}$.

\begin{rem}\label{rem:er-mse}
 It is notable that in this formulation where the cost is the error rate,  positive probability masses corresponding to higher noise outcomes tend to be less likely than those having smaller outcomes. This is why these designs exhibit low MSE.
\end{rem}
Mechanisms with better error rate (lower $\rho$)  must allow for $\delta > 0$. It can be proven (see Theorem~\ref{Thm:Linear_Flat_Region}) that the optimal error rate $\rho^{\star}(\delta,\epsilon)$ vs. $\delta$ curve is piece-wise linear, interleaving flat regions with intervals with linear negative slope, see  Fig.~\ref{fig:flat_linear_region}.
\begin{figure}
    \centering
    \includegraphics[width=0.45\textwidth]{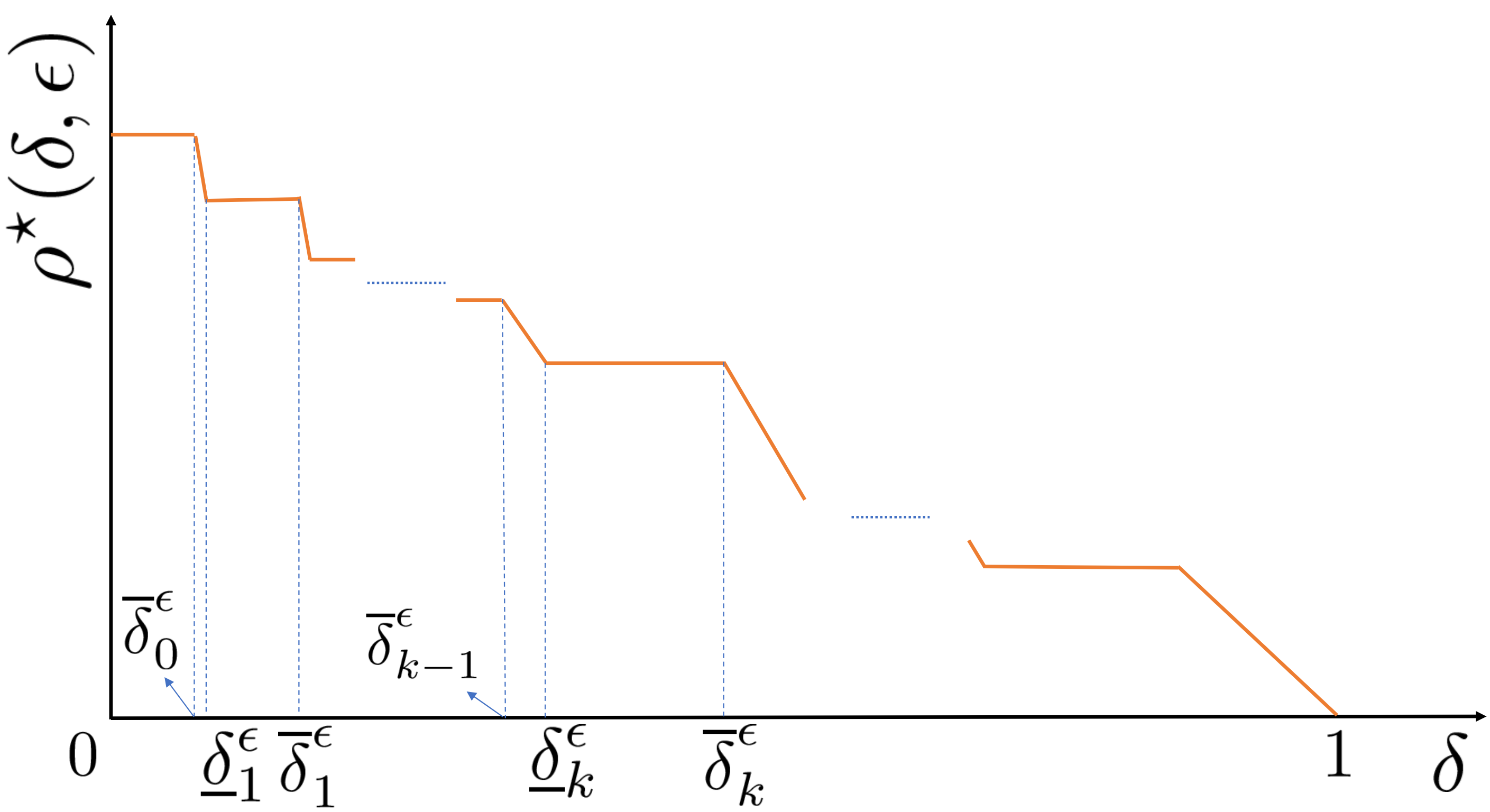}
    \caption{The variation of $\rho$ as a function of $\delta$ for SD neighborhood showing the alternate flat and linear regions.}
    \label{fig:flat_linear_region}
\end{figure}
We categorize them as ``linear regions'' and ``flat regions''. 
Let $\underline{\delta}^{\epsilon}_k$ and $\overline{\delta}^{\epsilon}_k$ be the instances of $\delta$ at which $\rho^{\star}(\delta,\epsilon)$ changes from $k^{th}$ linear region to $k^{th}$ flat region and $k^{th}$ flat region to $(k+1)^{th}$ linear region, respectively.

\begin{rem} \label{Rem:PrimeRemark} We state the  following results of this section, i.e., Section~\ref{Sec:Opt_Noise_SD_Nbd}, for the case $(n+1, \hat{\mu})$ that are relatively prime. For the case where they are not, the results are obtained by replacing $n$ with $N_{\hat{\mu}} = \frac{(n+1)}{\gcd((n+1),\hat{\mu})}$ in all the expressions and explanations.
\end{rem}

Let us define the following quantities:
\begin{subequations}
\label{eq:delta-top-bottom}
\begin{align} 
\overline{\delta}^{\epsilon}_k&:= e^{-(n-k-1)\epsilon}\frac{1-e^{-\epsilon}}{1-e^{-(n-k+1)\epsilon}},  ~\mbox{for $k\in [n\!-\!1]$}\\
\overline{\delta}^{\epsilon}_n&:=1,
~~~~
\underline{\delta}^{\epsilon}_0:=0
\\
\underline{\delta}^{\epsilon}_k&:=
e^{-(n-k)\epsilon}\frac{1-e^{-\epsilon}}{1-e^{-(n-k+1)\epsilon}}, ~\mbox{for $k\in [n]_+$}.
\end{align}
\end{subequations}
\begin{thm}\label{Thm:Linear_Flat_Region}
$k\in [n]_+$, in the $k^{th}$ section where $\rho^\star(\delta,\epsilon)$ is flat  within 
$\underline{\delta}^{\epsilon}_{k} \leq \delta\leq \overline{\delta}^{\epsilon}_k$:
\begin{align}\label{eq:SD_f*(h)flat}
f^\star_{(h)}&=
\begin{cases}
    \underline{\delta}^{\epsilon}_{k} e^{(n-k-h)\epsilon}, &h \in [n-k],\\
    0, &h \in [n-k+1:n],
\end{cases}
\end{align}
In the portion of the $k^{th}$ section where $\rho^\star(\delta,\epsilon)$ decreases linearly with $\delta$, which are within $~ \overline{\delta}^{\epsilon}_{k-1}< \delta\leq \underline{\delta}^{\epsilon}_k$:
\begin{align}\label{eq:SD_f*(h)linear}
f^\star_{(h)}=
\begin{cases}
    \delta e^{(n-k-h)\epsilon}, &\!\!\!\!h \!\in \! [n\!-\!k],\\
    e^{(n-h)\epsilon}\frac{e^{\epsilon}-1}{e^{k\epsilon}-1}\left(
    1-\frac{\delta}{\underline{\delta}^{\epsilon}_k}
    \right),& \!\!\!\!h\!\in \! [n\!-\!k+1\!:\!n].
\end{cases}
\end{align}
and $\rho^\star(\delta,\epsilon) =1-f^\star_{(0)}$. 
\end{thm}
\begin{proof}
The proof is in Appendix~\ref{Apdx:Thm3Proof}.
\end{proof}
From Theorem~\ref{Thm:Linear_Flat_Region} we note that the boundary point $\underline{\delta}^{\epsilon}_k$ determines the value of $f^\star_{(n-k)}$, which is the smallest non-zero probability mass in the $k^{th}$ flat region. Similarly, the other boundary point $\overline{\delta}^{\epsilon}_k \equiv e^{\epsilon} \underline{\delta}^{\epsilon}_k$ indicates the value of $f^\star_{(n-k-1)}$, which is the smallest non-zero probability mass in the $(k-1)^{th}$ flat region.
Having calculated the optimal PMF for the SD neighborhood case in Theorem~\ref{Thm:Linear_Flat_Region}, the $(\epsilon,\delta)$ curves correspond to $f^{\star}(0)=1-\rho^{\star}$ for all its $n+1$ possible expressions or, better stated, they are the level curves  $\rho(\delta, \epsilon)=\rho^{\star}$. The trend of $\delta$ curves is monotonically decreasing with respect to $\epsilon$ for a given $\rho^{\star}$. 
Let $\epsilon_0^{\rho^{\star}}$ be the solution obtained setting $f^{\star}(0)$ in \eqref{eq:SD_f*(0)} to be equal to $1-\rho^{\star}$:
\begin{align}\label{eq:eps0}
\epsilon_0^{\rho^{\star}}:~~\rho^{\star}=1-\frac{1-e^{-\epsilon_0^{\rho^{\star}}}}{1-e^{-(n+1)\epsilon_0^{\rho^{\star}}}} 
\end{align}
Then we must have that $\delta=0$ for
$\epsilon\geq \epsilon_0^{\rho^{\star}}$. Because this corresponds to a flat region for $f^{\star}(0)$, there has to be a discontinuity moving towards lower values $\epsilon<\epsilon_0^{\rho^{\star}}$, and $\delta$ must immediately jump to $\overline{\delta}^{\epsilon_0^{\rho^{\star}}}_0$ when $\epsilon$ is an infinitesimal amount below $\epsilon_0^{\rho^{\star}}$.
This point is the edge of the {\it linear region}. For a range of values of $\epsilon<\epsilon_0^{\rho^{\star}}$, $\delta$ with respect to $\epsilon$ must have the negative exponential trend $\delta =(1-\rho^{\star})e^{-n\epsilon}$ obtained from equation \eqref{eq:SD_f*(h)linear} for $h=0,k=1$ until the next jump occurs, for a value $\epsilon_1^{\rho^{\star}}$ which is obtained by setting $f^{\star}(0)$ for $h=0$ and $k=1$ in \eqref{eq:SD_f*(h)flat} to be equal to $1-\rho^{\star}$:
\begin{align}
    \epsilon_1^{\rho^{\star}}:~~\rho^{\star}=1-\underline{\delta}_1^{\epsilon_1^{\rho^{\star}}}
    e^{(n-1)\epsilon_1^{\rho^{\star}}}=1-\frac{1-e^{-\epsilon_1^{\rho^{\star}}}}{1-e^{-n\epsilon_1^{\rho^{\star}}}}.
\end{align}
Following this logic, one can prove that the optimum $(\epsilon, \delta)$ curve for a given error rate $\rho^{\star}$ is:
\begin{cor}
\label{Cor:eps_delta_SD}
For a given $\epsilon>0$, the privacy loss for the SD neighborhood case with the optimal noise mechanism, is a discontinuous function of $\epsilon$, where:
\begin{align}
    \delta^\epsilon = e^{-(n-k)\epsilon}(1-\rho^{\star}),~~~\epsilon^{\rho^{\star}}_{k} \leq \epsilon<\epsilon^{\rho^{\star}}_{k-1} 
\end{align}
when $\rho^\star = 1 - f^\star_{(0)}$ is in $k^{th}$ section, $k \in [n]_+$ and $\epsilon_k^{\rho^{\star}}$ are the solutions of the following equations:
\begin{equation}\label{Eq:eps-rho*}
        \epsilon_k^{\rho^{\star}}\!\!:\!\rho^{\star}=1-\underline{\delta}_k^{\epsilon_k^{\rho^{\star}}}
    e^{(n-k)\epsilon_k^{\rho^{\star}}}=1-\frac{1-e^{-\epsilon_k^{\rho^{\star}}}}{1-e^{-(n+1-k)\epsilon_k^{\rho^{\star}}}}.
\end{equation}
\end{cor}
\begin{proof}
As discussed before the Corollary, for $\delta=0$, the level curves of $\rho^{\star}(\delta,\epsilon)=\rho^{\star}$ as a function of $\epsilon$ must be monotonically decreasing for $\epsilon\geq \epsilon_0^{\rho^{\star}}$ (see \eqref{eq:eps0}) as $\epsilon$ increases. Then the curve will have discontinuities that correspond to the flat regions and the trend between these discontinuities is obtained through the equation $(1-\rho^{\star})=f^{\star}(0)=\delta e^{(n-k)\epsilon}$, which implies $\delta=(1-\rho^{\star})e^{-(n-k)\epsilon}$. The $k^{th}$ interval starts at the point $\epsilon=\epsilon_k^{\rho^{\star}}$ such that $\delta=\underline{\delta}_k^{\epsilon=\epsilon_k^{\rho^{\star}}}$  ensures that $f^{\star}(0)=1-\rho^{\star}$; thus, $\epsilon=\epsilon_k^{\rho^{\star}}$ must be the solution of \eqref{Eq:eps-rho*}.   
\end{proof}
\subsubsection{The BD neighborhood} \label{Sec:Opt_Noise_BD_Nbd}

Also in this case we first start with the optimum noise for $\delta = 0$. First, let us express $n$ as:
\begin{equation}
\label{eq:n_represent}
  n = b\bar{\mu}+r  
\end{equation}
where $b = \lfloor{n/\bar{\mu}}\rfloor$
and $r \in [\bar{\mu}-1]$.
For this case, the inequalities in~\eqref{Eq:f0_constraint} can be written as the following:
\begin{align}
    \label{Eq:f_eta_ordered_BD}
    f^\star{(h)} - e^\epsilon f^\star{(h+\mu)}\leq  0,~~\forall h \in [n],  ~\forall \mu \in [\bar{\mu}]_+.
\end{align}
\begin{figure}[!ht]
\subfloat[][]{%
  \includegraphics[clip,width=0.45\textwidth]{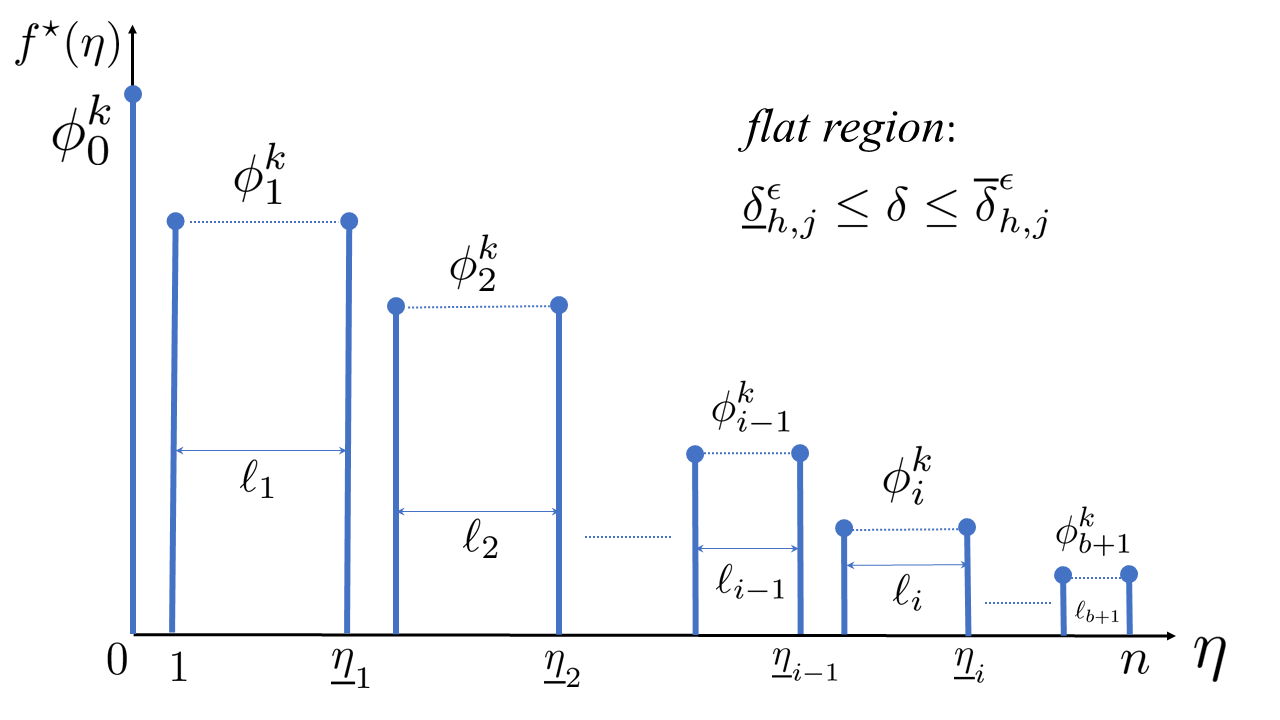}%
  \label{fig:flat_PMF_BD}
}\\
\subfloat[][]{%
  \includegraphics[clip,width=0.45\textwidth]{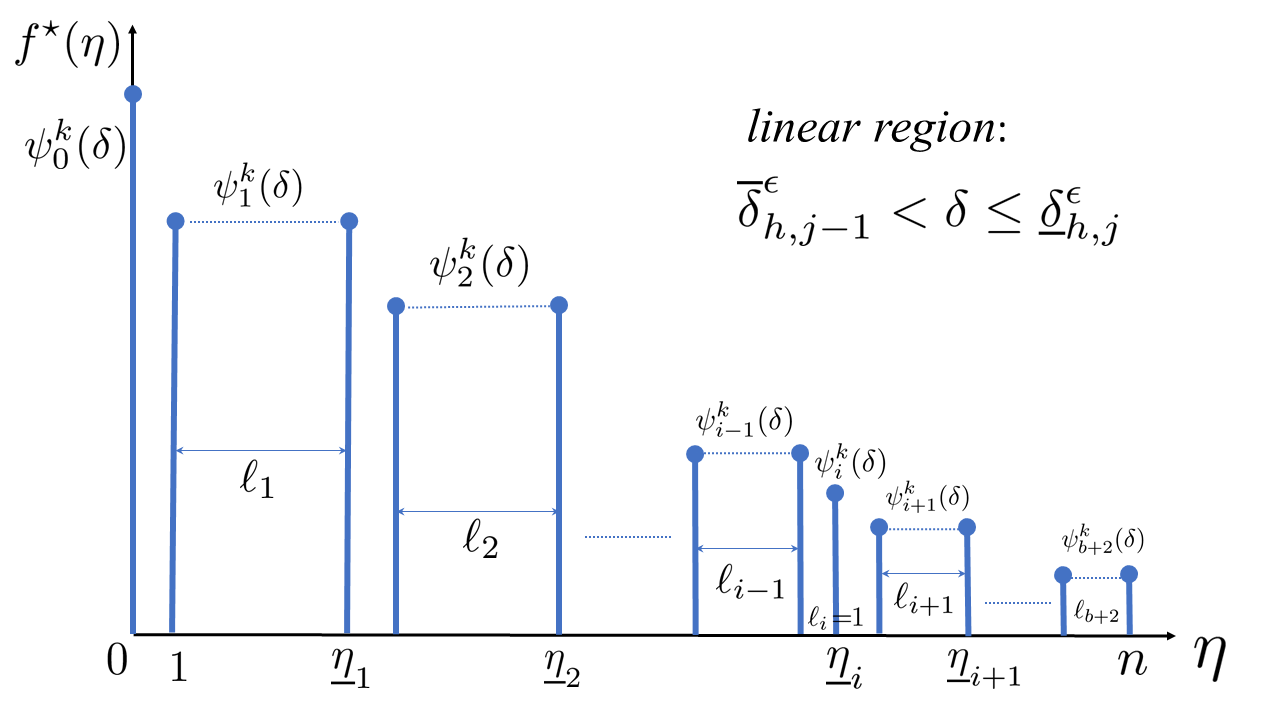}%
  \label{fig:linear_PMF_BD}
}
\caption{The PMF of the optimal noise mechanism for the BD neighborhood follows a staircase pattern. In the flat region, it has $b+1$ steps and $i^{th}$ step height is $\phi^k_i$. Similarly, in the linear  region, the PMF has $b+2$ steps and $i^{th}$ step height is $\psi^k_i(\delta)$.}
\label{fig:FixedCardNoise}
\end{figure}
\begin{lem}\label{Lem:BD_delta0}
In the case of a BD neighborhood of size $\bar{\mu}$, the optimum PMF $f^\star(\eta)$ has the following form for $i\in[b+1]_+$:
\begin{equation}
    f^\star(\eta)=f^\star(0)e^{-i\epsilon}\!\equiv \phi_i,
    ~~(i\!-\!1)\bar{\mu}+1\leq \eta\leq \min(i\bar{\mu},n).
    \label{eq:BD_fstarValues}
\end{equation}
and the mass at zero is:
\begin{align}
    f^\star(0)&=\left(1+\bar{\mu}\sum_{i=1}^b e^{-i\epsilon}+
    re^{-(b+1)\epsilon}\right)^{-1}\equiv \phi_0.\label{eq:Noise_FixedCardi}
\end{align}
The PMF has, therefore, a staircase trend with steps of length $\ell_i=\bar{\mu}$ for $i\in[b]_+$ and $\ell_{b+1}=r$ and $\rho^\star=1-\phi_0$. 
\end{lem}
\begin{proof}
The proof is similar to that of Lemma \ref{Lem:SD_delta0} because it recognizes that 
it is best to meet the inequalities in \eqref{Eq:f_eta_ordered_BD} as equalities, since that allows for the largest $f^{\star}(0)$. The only difference is that the masses in the first group $[\bar{\mu}]_+$ are equal to $\phi_1=e^{-\epsilon}f^{\star}(0)$, thus they constrain a second group to have value $\phi_2=e^{-\epsilon}\phi_1$ and so on. There are $b$ of them that contain $\bar{\mu}$ masses of probability and only the last group includes the last $r$ values. $f^\star(0)$ is obtained normalizing the PMF to add up to 1. 
\end{proof}

The PMF for $\delta>0$ has staircase pattern (see Fig.~\ref{fig:FixedCardNoise}), similar to Lemma \ref{Lem:BD_delta0} and, also in this case, for $\delta>0$  the $\rho^\star(\delta,\epsilon)$ has a piece-wise linear trend that alternates between flat and linear regions. However, the BD neighborhood case has an intricate pattern in which the constraints become violations, as the privacy loss $\delta \rightarrow 1$. Rather than having $k=n$, the number of sections $k$ is quadratic with respect to $b$. To explain the trend, it is best to divide the section of the $f^{\star}(0) \equiv 1 - \rho^\star(\delta,\epsilon)$ curve vs. $\delta$ in $b$ segments, indexed by $h \in [b]_+$ as shown in Fig.~\ref{fig:BD_f0_vs_delta}. Except for the last interval corresponding to $h=b$, the other segments, indexed by $h \in [b-1]_+$, are further divided into $h$ segments, indexed by $j \in [h]_+$ and this index refers to one of the alternating flat and linear regions within the $h^{th}$ interval. This results in  $k=\frac{b(b-1)}{2}$ alternating flat and linear regions. Instead, in the segment indexed by $h=b$, there is only one linear region i.e.,  $f^{\star}(0) = \delta$. The optimum distribution is specified in the following theorem:

\begin{figure}
    \centering
    \includegraphics[width=0.45\textwidth]{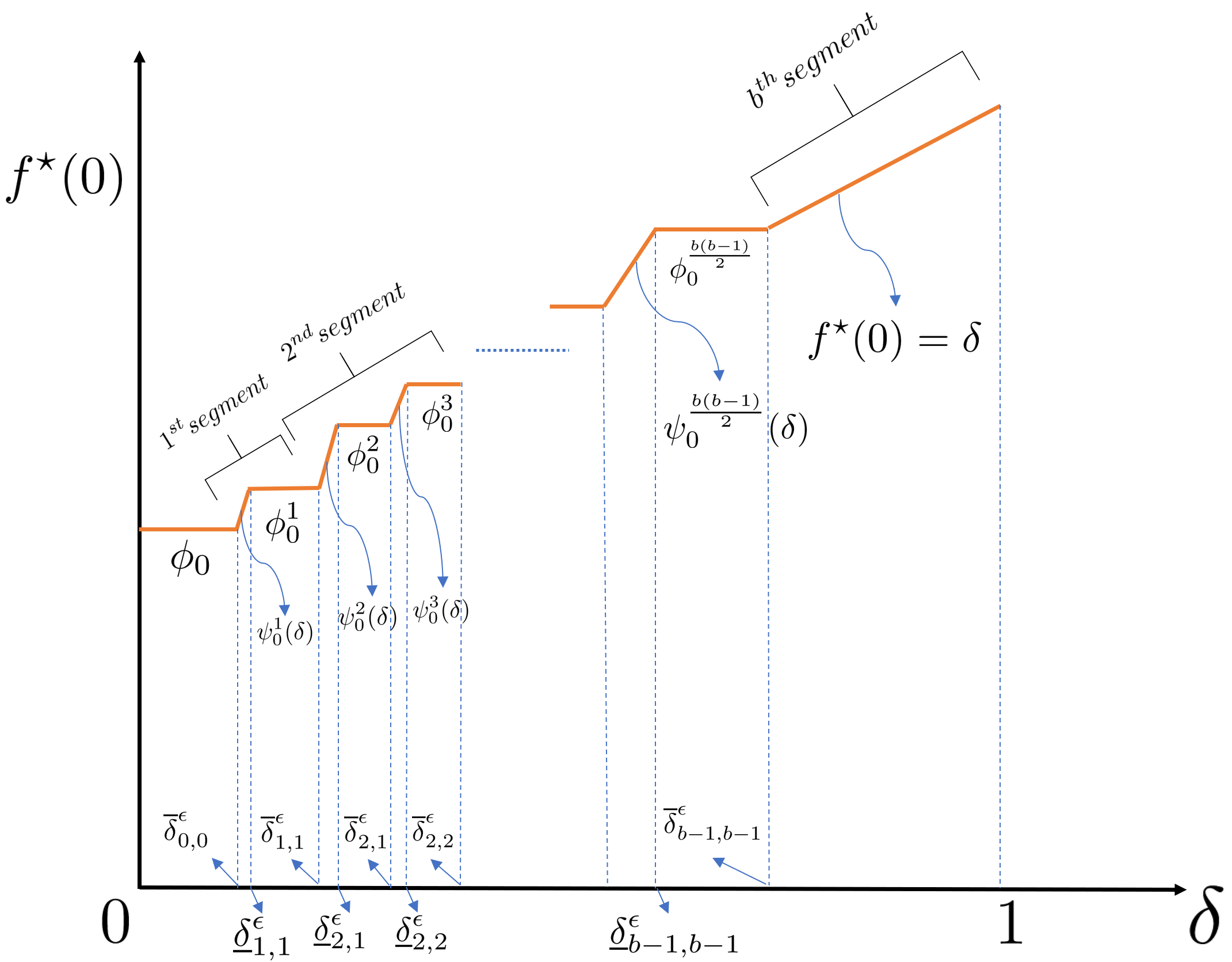}
    \caption{The variation of $f^{\star}(0)$ as a function of $\delta$ for BD neighborhood showing $b$ segments with the alternate flat and linear regions.}
    \label{fig:BD_f0_vs_delta}
\end{figure}
\begin{thm}
\label{Thm:BD_results}
 Let $b+r<\bar{\mu}$. For a given $\epsilon >0$ and for $\mu_{\XX}  \in [\bar{\mu}]_+$, $~\forall X \in {\mathcal X},\forall X'\in {\mathcal X}_{X}^{(1)}$ (i.e. the BD neighborhood), $f^{\star}(0)$  versus $\delta$ features flat and linear regions as shown in Fig.~\ref{fig:BD_f0_vs_delta}. In the first $(b-1)$ segments, indexed by $h\in [b-1]_+$, each alternating a pair of flat and linear regions indexed $j \in [h]_+$, with respective boundaries $\underline{\delta}^{\epsilon}_{h,j}\leq \delta\leq \overline{\delta}^{\epsilon}_{h,j}$ and $\overline{\delta}^{\epsilon}_{h,j-1} < \delta\leq \underline{\delta}^{\epsilon}_{h,j}$, with the convention $\overline{\delta}^{\epsilon}_{h,0}=\overline{\delta}^{\epsilon}_{h-1,h-1}$.
The following facts are true:\\[0.1cm]
 {\bf (a)}   In the $k^{th}$ flat region, $k=\sum_{j'=1}^{h-1} j'+j=\frac{(h-1)h}{2}+j$, the optimum PMF (c.f. Fig.~\ref{fig:flat_PMF_BD}) for $i\in [b+1]_+$ is:
\begin{equation}
    \label{eq:PMF_BD_thm}
    f^\star(\eta)= f^\star(0)e^{-i\epsilon}\equiv \phi_i^k,~ \underline{\eta}_{i-1}< \eta \leq \underline{\eta}_{i},\underline{\eta}_{0}=0,
\end{equation}
where what distinguishes the distributions for each $k$ are the intervals $\underline{\eta}_{i-1}< \eta \leq \underline{\eta}_{i},i \in [b+1]_+$ with equal probability mass $\phi^k_i$. More specifically, considering the $k^{th}$ flat region, corresponding to the pair $h,j$ with $h\in [b-1]_+$, $j\in [h]_+$,
the intervals $\underline{\eta}_{i-1}< \eta \leq \underline{\eta}_{i}$ of the optimum PMF,  for $i\in [b+1]_+$, have lengths $\ell_i=\underline{\eta}_{i}-\underline{\eta}_{i-1}$:
\begin{align}
   \ell_i&\!=\!
   \begin{cases}
          1,& \!\!\!\!\mbox{for}~ i=0\\
          \bar{\mu},& \!\!\!\!\mbox{for} ~i \in [b\!-\!h]_+ \\
     \bar{\mu}\!-\!1,& \!\!\!\!\mbox{for}~i \in [b\!-\!h\!+\!1\!:\!b], \\&~~~i \! \neq \!b\!-\!h\!+\!j\!+\!1\\
     \bar{\mu},& \!\!\!\!\mbox{for} ~i \!=\!b\!-\!h\!+\!j\!+\!1 \mbox{ when } \!j\! \neq \!h \\
     r+h-u_{hj},& \!\!\!\!\mbox{for}~i=b+1.
   \end{cases}
\\
    u_{hj} &:= 
   \begin{cases}
    1, ~& \mbox{for}~ j \neq h,\\
    0, ~& \mbox{for}~ j=h.
\end{cases}
\end{align}
The corresponding indexes sets are obtained as:
\begin{equation}
    \underline{\eta}_{0}=0,  ~~\underline{\eta}_{i}=\underline{\eta}_{i-1}+\ell_{i},~~i\in[b+1]_+.
\end{equation}
To normalize the distribution $f^\star(0)= \phi^k_0$ must be:
\begin{align}
    \phi^k_0&=\frac{1}{\sum_{i=0}^{b+1}\ell_i e^{-i\epsilon}} =\left(\bar{\mu}\alpha_{hj}^{\epsilon} \!+\!\beta_{hj}^{\epsilon}\right)^{-1}, \label{eq:phi_0}\\
    \alpha^{\epsilon}_{hj}&:=\! e^{-\epsilon} \xi_{b}^{\epsilon} \!-\! (1\!-\!u_{hj})e^{-\!(b-h+j+1)\epsilon}, \\
    &~~~~~~\text{where }\xi_{a}^{\epsilon} :=  \sum_{i'=0}^{a-1} e^{-i'\epsilon}, \forall a\in \mathbb{N}^+ \nonumber\\
    \beta_{hj}^{\epsilon} &:=1+ e^{-(b-h+1)\epsilon}(e^{-j\epsilon} - \xi_{h}^{\epsilon}) \nonumber \\
    &~~~~~~+(r+h-u_{hj}) e^{-(b+1)\epsilon} 
\end{align}
and the PMF is valid within $\underline{\delta}^{\epsilon}_{h,j}\leq \delta\leq \overline{\delta}^{\epsilon}_{h,j}$ where:
\begin{subequations}
\begin{align}
    \overline{\delta}^{\epsilon}_{h,j}&= \phi_{0}^{k} e^{-(b-h)\epsilon} \sum\limits_{j'=0}^{j} e^{-j'\epsilon}\nonumber\\
    &=\phi_{0}^{k} e^{-(b-h)\epsilon}\xi^{\epsilon}_{j+1},j\in [h-1]\\ 
    \overline{\delta}^{\epsilon}_{h,h}&=\phi_{0}^{k} e^{-(b-h-1)\epsilon}, ~\overline{\delta}^{\epsilon}_{0,0}=\phi_{0} e^{-(b-1)\epsilon},  \label{eq:delta_overline_BD} \\
    \underline{\delta}^{\epsilon}_{h,j}&= \phi_{0}^{k}e^{-(b-h)\epsilon} \xi_{j}^{\epsilon} \equiv \left(\frac{\phi_{0}^{k}}{\phi_{0}^{k-1}}\right) \overline{\delta}^{\epsilon}_{h,j-1} \label{eq:delta_underline_BD}
\end{align}
\end{subequations}
{\bf (b)} In each of the linear regions, i.e.,  $\overline{\delta}^{\epsilon}_{h,j-1} < \delta\leq \underline{\delta}^{\epsilon}_{h,j}$, the PMF values vary linearly in groups with respect to $\delta$. The group lengths are:
\begin{align}
   \ell_i&=
   \begin{cases}
        1,& \mbox{for}~i=0\\
     \bar{\mu},& \mbox{for}~i \in [b-h]_+\\
     \bar{\mu}-1,& \mbox{for}~i \in [b-h+1:b], \\&~~~~~i \! \neq b-h+j+1\\ 
     1,& \mbox{for}~i=b-h+j+1\\
     r+h-1,& \mbox{for}~i=b+2\\
   \end{cases}
\end{align}
each with probability $\psi_i^k(\delta)$ (as shown in Fig.~\ref{fig:linear_PMF_BD}):
\begin{align}
&\psi_i^k(\delta) = 
 \begin{cases}
{\delta e^{(b-h-i) \epsilon}}/{\xi_{j}^{\epsilon}}, &\!\!\!\!i \in  [b-h+j] \\
{\delta e^{(b-h-i+1) \epsilon}}/{\xi_{j}^{\epsilon}}, &\!\!\!\!i \!\in \! [b\!-\!h\!+\!j\!+\!2:\!b+2],
\end{cases}\nonumber\\
&\psi_i^{b-h+j+1}(\delta)  \!=  1 -\!\!\!\!\!\! \sum\limits_{\substack{i=0\\i \neq b-h+j+1}}^{b+2}\!\!\! \ell_i \psi_i^k(\delta).
\label{eq:step_height_thm}
\end{align}
{\bf (c)} In the $b^{th}$ segment, i.e. $\overline{\delta}^{\epsilon}_{b-1,b-1} < \delta \leq 1$, the objective function maximum $f^\star(0)=\delta$. So any set of $f^\star(\eta), ~\eta \in [n]$, that satisfy $\sum\limits_{\eta=1}^{n} f^\star(\eta) = 1 - \delta$ and $u_{\XX} (0) = 1, u_{\XX} (\eta) = 0, ~\eta \in [n]_+$, provides the optimal PMF. One of the possible solutions is:
\begin{align}
    f^\star(0) &= \delta, \\
    f^\star(\eta) &= \frac{1-\delta}{n}, ~~\eta \in [n]_+.
\end{align}
Note that, in the last segment, the optimal PMF no longer follows the staircase pattern. 
\end{thm}

\begin{proof}
The proof is in Appendix~\ref{Apdx:Thm4Proof}.
\end{proof}
In Theorem~\ref{Thm:BD_results} we cover the case $b+r<\bar{\mu}$. When $b+r \geq \bar{\mu}$ the result is not conceptually difficult, but the optimal PMF is hard to express in a readable form. We discuss the general case towards the end of Appendix~\ref{Apdx:Thm4Proof}.

\begin{cor}
\label{Cor:eps_delta_BD}
For a given $\epsilon>0$, the privacy loss for the BD neighborhood case with the optimal noise mechanism is also a discontinuous function of $\epsilon$, where:
\begin{align}
    \delta^\epsilon = \xi_{j}^{\epsilon} e^{-(b-h)\epsilon}(1-\rho^{\star}),~~~\epsilon^{\rho^{\star}}_{h,j} \leq \epsilon<\epsilon^{\rho^{\star}}_{h,j-1} 
\end{align}
when $\rho^\star = 1 - \psi_0^k(\delta)$ is in $k^{th}$ section, $k \in \left[\frac{b(b-1)}{2}\right]_+$ and $\epsilon_{h,j}^{\rho^{\star}}$, $h \in [b-1]_+,~j \in [h]_+$, are the solutions of:
\begin{align}\label{Eq:eps-rho_BD*}
        \epsilon_{h,j}^{\rho^{\star}}:~~\rho^{\star}=1-\phi_0^k=1-\left(\bar{\mu}\alpha_b^{\epsilon^{\rho^{\star}}_{h,j}} \!+\!\beta_{hj}^{\epsilon^{\rho^{\star}}_{h,j}}\right)^{-1}.
\end{align}
\end{cor}
\begin{proof}
The proof is a direct extension of that for Corollary~\ref{Cor:eps_delta_SD} and is omitted for brevity.
\end{proof}

\subsubsection{Optimal Error Rates as $n \rightarrow \infty$} \label{Sec:InfinityCase}
In this section, we study the limit for $n\rightarrow \infty$ of the distributions for the two cases we studied, the SD and BD neighborhoods.
First, we discuss the $\delta=0$ case.
The goal is to find the relationship between $\epsilon$ and $\rho$ when n $\rightarrow \infty$.  From Lemma \ref{Lem:SD_delta0}~\eqref{eq:SD_f*(0)}, we see that for SD neighborhood case, $f^{\star}(0) \equiv 1 - \rho^{\star}(\epsilon) \rightarrow 1 - e^{-\epsilon} \implies \rho^{\star}(\epsilon) \rightarrow e^{-\epsilon}$ as $n \rightarrow \infty$. Since, $\rho^{\star}(\epsilon)$ is also constrained by 0.5, we have that the limit function $ \rho^{\star}_{\infty}(\epsilon)$:
\begin{align}
    \rho^{\star}_{\infty}(\epsilon) &=
    \begin{cases}
         0.5, ~\epsilon \in (0, \ln{2}] \\
         e^{-\epsilon}, ~\epsilon \geq \ln{2}.
    \end{cases}
\end{align}
And the optimum PMF is zero for all $\eta\neq h\hat{\mu}$,  and:
\begin{align}
f_{\infty}^\star(h\hat{\mu})=(1-\rho^{\star}_{\infty}(\epsilon))e^{-h\epsilon}, ~~h\in \mathbb{N}
\end{align}
Similarly, for the BD neighborhood case and $
\delta=0$ in Lemma \ref{Lem:BD_delta0} as $n \rightarrow \infty \implies b \rightarrow \infty$,  from~\eqref{eq:Noise_FixedCardi} we get: 
\begin{align}
    \phi_0 \equiv 1 - \rho^{\star}(\epsilon) &\rightarrow \left(1+\frac{\bar{\mu}e^{-\epsilon}}{ 1-e^{-\epsilon}}\right)^{-1} \label{eq:phi0}\\
    \implies \rho^{\star}(\epsilon) &\rightarrow \frac{\bar{\mu}e^{-\epsilon}}{\bar{\mu}e^{-\epsilon} + (1-e^{-\epsilon})}
\end{align}
Hence, the expression for $\rho^{\star}(\epsilon)$ for any $\epsilon >0$ is:
\begin{align}
    \rho^{\star}_{\infty}(\epsilon) &=
    \begin{cases}
         0.5, &~\epsilon \in (0, \ln{(1+\bar{\mu})}] \\         \frac{\bar{\mu}e^{-\epsilon}}{\bar{\mu}e^{-\epsilon} + (1-e^{-\epsilon})}, &~\epsilon \geq \ln{(1+\bar{\mu})}.
    \end{cases}
\end{align}
Each of the PMF staircase steps becomes of size $\bar{\mu}$ and the values have an exponentially decaying trend:
\begin{align}
f_{\infty}^\star(0)&=1-\rho^\star_{\infty}(\epsilon)\\
    f^{\star}_{\infty}(\eta)&=f_{\infty}^\star(0)e^{-h\epsilon},~(h\!-\!1)\bar{\mu}<\eta\leq h\bar{\mu}, ~h\in \mathbb{N}^+.
\end{align}
For $0<\delta \leq 1$, it is convenient to use the index $i=n-k$ looking at the trend of the distortion from $\delta=1$, where $f^{\star}(0)=1$ backward. Because the discontinuities between flat and linear regions happen at the points where $\delta=\underline{\delta}^{\epsilon}_{n-i},~i\in [n-1]$ we can see from Theorem \ref{Thm:Linear_Flat_Region} the distortion for $i\in [n-1]$:
\begin{equation}
    f_i^\star(0)\leq \frac{1\!-\!e^{-\epsilon}}{1\!-\!e^{-(i+1)\epsilon}}\Rightarrow \rho^\star(\delta,\epsilon)\geq 1- \frac{1\!-\!e^{-\epsilon}}{1\!-\!e^{-(i+1)\epsilon}}, 
\end{equation}
and the size of the intervals shrinks like an $o(e^{-i\epsilon})$, as $i\rightarrow +\infty$ quickly leading to the same result as $\delta\rightarrow 0$, where the distortion tends to $e^{-\epsilon}$ as stated before. 

Similarly, for the BD neighborhood case, to  find the expressions for $\phi_0^{\infty}$ and  $\phi_i^{\infty}$, it is convenient to use a new index $c=b-h$, looking at the trends of the distortion from $\delta=1$, where $f^{\star}(0)=1$, going backwards towards $\delta=0$. In the $b^{th}$ segment (part (c) of Theorem~\ref{Thm:BD_results}, $f_b^\star(\eta) \rightarrow 0$ as $b \rightarrow \infty$ and $n \rightarrow \infty$ for $\eta \in [n]_+$ and thus $f_b^\star(0) \rightarrow 1$. 
For $\delta \approx 1$, in the $c^{th}$ region we get the following expressions by using~\eqref{eq:phi_0}:
\begin{align}
    \alpha^{\epsilon}_{b-c,j} & \!\rightarrow\! \frac{e^{-\epsilon}}{1-e^{-\epsilon}};
    ~~\beta^{\epsilon}_{b-c,j} \!\rightarrow\! 1 - \frac{e^{-(c+j+1)\epsilon}}{1-e^{-\epsilon}}
\end{align}
\begin{align}
    \implies \!\!
    \phi_0^{\infty}(c,j) &\rightarrow \frac{1-e^{-\epsilon}}{\bar{\mu}e^{-\epsilon} \!+\! (1\!-\!e^{-\epsilon}(1\!+\!e^{-(c+j)\epsilon}))}, \label{eq:phi0_infity}\\
    \phi_i^{\infty} &\rightarrow e^{-i\epsilon} \phi_0^{\infty}, ~~i \in [b+1]_+.
\end{align}
Now, as $c \rightarrow \infty$, the expression of $\phi_0^{\infty}(c,j)$ in~\eqref{eq:phi0_infity} converges to $\phi_0$ as shown in~\eqref{eq:phi0}, i.e. the result for $\delta \rightarrow 0$.

\subsubsection{Optimal Noise Mechanism for Vector Queries} \label{Sec:Opt_Noise_Vector}
Next, we briefly discuss the optimal noise mechanism design for vector queries { to explore what difference it makes to optimize after mapping each vector onto a number in $[n]$ vs. adopting the mechanism on an entry-by-entry basis.} In fact, let each entry of a vector query be in the set ${\mathcal Q}$. In the first case, the MILP formulation of the problem defined in~\eqref{eq:MILP_obj}--\eqref{eq:eta_range} can be applied directly to vectors of queries considering the masses of probabilities as a joint PMF, with arguments corresponding to all possible tuples in ${\mathcal Q}^k$.  In Section~\ref{Sec:Simulations}, we provide two examples for 2D vector queries-- one for the BD neighborhood and another for an arbitrary neighborhood (see Fig.~\ref{Fig:PMF_2D_plots}). We observe that for the BD neighborhood case, the optimum noise mechanism follows a staircase pattern in 2D as well and for the arbitrary neighborhood, the optimum noise mechanism has $e^{-\epsilon} f^{\star}(\bm 0)$ values at $\bm \eta = \bm \mu_{\XX}$, where boldface letters are vectors.  
\begin{rem} \label{Rem:VectorQuery}
 In general,  the optimal multidimensional noise mechanism does not amount to adding independent random noise to each query entry.
 In  Section~\ref{Sec:Simulations} we corroborate this statement by showing a counter-example obtained using the MILP program for the vector case, considering a two-dimensional vector query.
 \end{rem}

\section{Numerical Results}\label{Sec:Simulations}
\begin{figure}
     \centering
     \subfloat[][]{\includegraphics[width=0.23\textwidth]{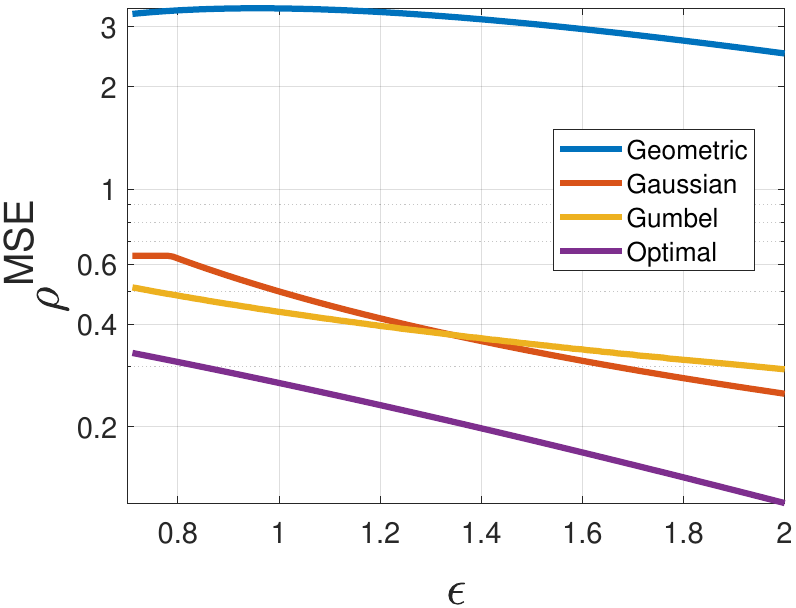}\label{fig:rho_MSECompare}}
     \subfloat[][]{\includegraphics[width=0.23\textwidth]{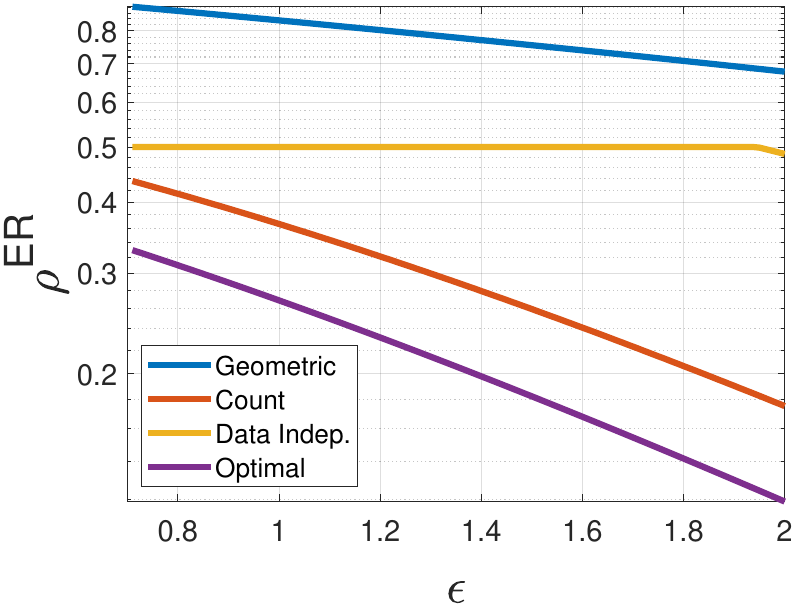}\label{fig:rho_ERCompare}}
     \caption{{Comparison of the proposed optimal mechanism in terms of the expected distortion costs with those proposed in~\cite{Ghosh_Universal,canonne2021discrete,sadeghi2020differentially,McSherry2007Mech,Ravi2022Diff}. In plot (a), the optimal noise mechanism is compared in terms of MSE, $\rho^{MSE}$, with the discrete geometric mechanism, discrete Gaussian mechanism, and Gumbel mechanism for a fixed a value of $\delta=0.3$ and $n=8$.} In plot (b), the optimal noise mechanism is compared in terms of error rate, $\rho^{ER}$, with the discrete geometric mechanism, discrete count mechanism, and data independent mechanism for a fixed a value of $\delta=0.5$ and $n=7$.}
     \label{Fig:MSE_ER_Comparisons}
\end{figure}

First, we compare the performance of our proposed optimal noise mechanism with the  discrete geometric mechanism~\cite{Ghosh_Universal}, discrete Gaussian mechanism~\cite{canonne2021discrete}, classical exponential mechanism~\cite{McSherry2007Mech}, discrete count mechanism~\cite{sadeghi2020differentially}, and data independent mechanism~\cite{Ravi2022Diff}. In plot~\ref{fig:rho_MSECompare}, the performance is compared in terms of $\rho^{MSE}$ vs. $\epsilon$ for a fixed value of $\delta=0.3$. Similarly, in plot~\ref{fig:rho_ERCompare}, the performance is compared in terms of $\rho^{ER}$ vs. $\epsilon$ for a fixed value of $\delta=0.5$.  
From the plots, it is clear that the proposed optimal noise mechanism significantly outperforms all these mechanisms.

Popular mechanisms for discrete queries are adding a random variable from a geometric distribution ~\cite{Ghosh_Universal},~\cite{Balcer2018DifferentialPO} or a quantized Gaussian distribution (see e.g. ~\cite{canonne2021discrete}). 
Clamping is an operation in which the query response $\tilde{q}$ is projected onto the domain  $[n]$ for any $\eta \in \mathbb{Z}$~\cite{Ghosh_Universal}, i.e.:
\begin{align}
    \tilde{q} = \min(\max(0, q+\eta),n). 
\end{align}
Let $F_{\eta}(\eta)$ denote the cumulative distribution function of $\eta$; then the distribution of $\tilde{{q}}$ in terms of the distribution of $\eta$ after clamping is as follows:
\begin{align} \label{eq:Noise_PMF_Clamp}
      f_{\tilde{q}}(k|q) = 
      & \begin{cases}
      F_{\eta}(-q), & ~~~~~~k = 0\\
      f_{\eta}(k-q), &~~~~~~k \in [n-1]_+\\ 1-F_{\eta}(n-q)
      & ~~~~~~k = n.
      \end{cases}      
\end{align}
From \eqref{eq:Noise_PMF_Clamp} one can compute the $(\epsilon, \delta)$ privacy curve  using~\eqref{Eq:PrivacyLoss},~\eqref{eq:epsilon_delta_relation} and~\eqref{eq:Noise_PMF_Clamp}. After clamping the $(\epsilon,\delta)$ guarantees provided by the said DP mechanisms are different from the ones calculated for $n=+\infty$ as shown in Fig.~\ref{fig:ClampingEffect}, reported for the same MSE = 3.38. 
From the figure, it is clear that clamping increases $\delta$ for the same $\epsilon$ budget, and this effect is particularly pronounced in the case of the Gaussian mechanism. 
Hence, it is not advisable to use infinite support-based noise mechanisms, such as discrete geometric and discrete Gaussian, in tandem with clamping operations to publish the discrete query response with finite support.\footnote{{The clamped Geometric mechanism has only one predefined parameter, $\alpha$, from which one can choose, whereas we choose all probability masses, so we do not have a single distribution, but we define a class of optimal distributions as the output of the optimization problem.}}
\begin{figure}
     \centering     
     \includegraphics[width=0.41\textwidth]{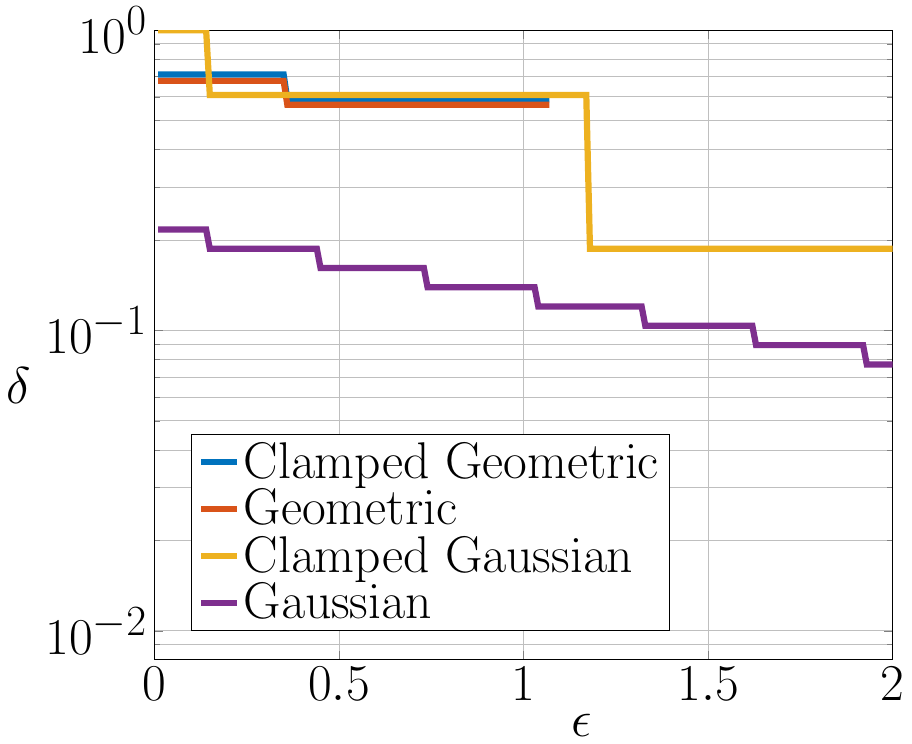}
     \caption{{This plot shows the adverse effect of clamping operation on the $(\epsilon,\delta)$ trade-off for discrete geometric and discrete Gaussian mechanisms. The following parameters are used: $n=8, \alpha = 0.7$ and $\sigma^2 = 3.38$.}}     \label{fig:ClampingEffect}
\end{figure}

Next, we now test the modified MILP problem  \eqref{eq:MILP_obj2} of minimizing $\delta$, constraining the error rate $\rho = \rho^{ER}$ 
and solve the MILP numerically using Gurobi as the solver. 
For a fair comparison, we consider the measure of errors vs. the ER and MSE, and respective parameters of the noise mechanisms viz., $\alpha$ in~\cite{Ghosh_Universal}, $\sigma^2$ in~\cite{canonne2021discrete}, $\beta'$ in~\cite{Adams2013Gumbel}\footnote{{To realize the classical exponential mechanism based DP~\cite{McSherry2007Mech}, we utilize the addition of $Gumbel(\epsilon,\beta')$ distribution to the query inputs~\cite{Adams2013Gumbel}, where $\beta'$ is a parameter.}}, and  $\rho$ in~\cite{sadeghi2020differentially,Ravi2022Diff}; the results are shown in Fig.~\ref{Fig:Comparisons}. 
\begin{figure}
     \centering
     \subfloat[][]{\includegraphics[width=0.23\textwidth]{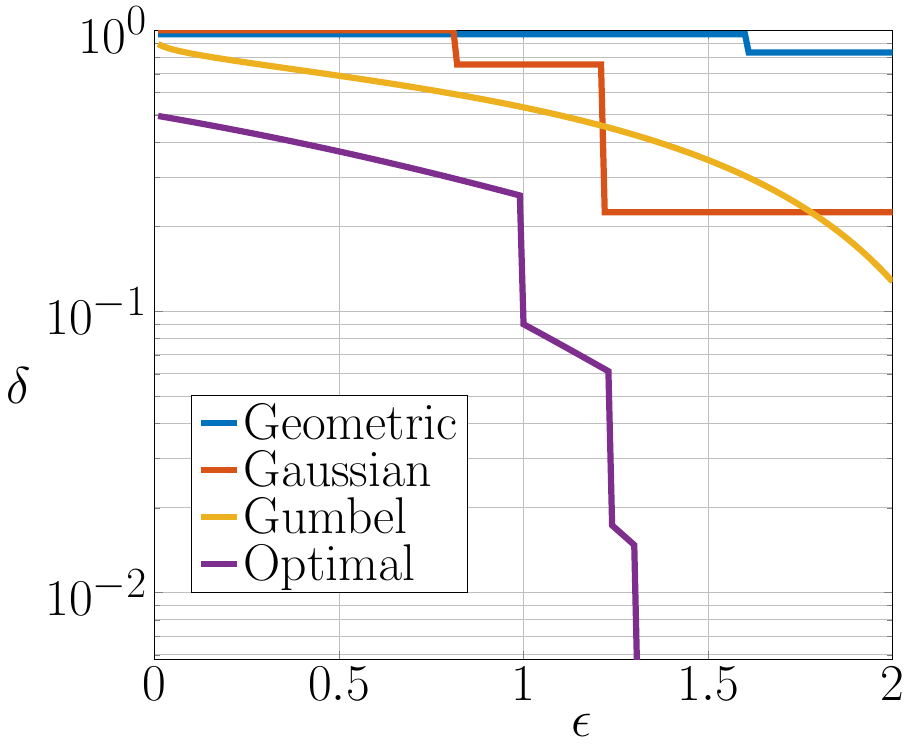}\label{fig:GeometricGaussianCompare}}
     \subfloat[][]{\includegraphics[width=0.23\textwidth]{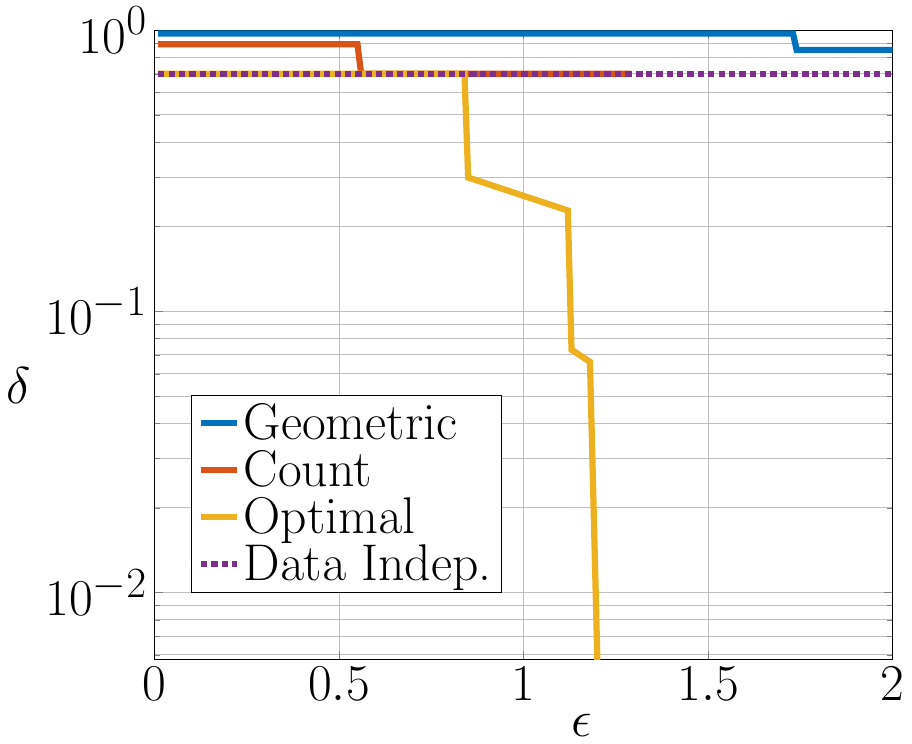}\label{fig:GeometricCountCompare}}
     \caption{{Comparison of the proposed optimal mechanism in terms of $(\epsilon, \delta)$ trade-offs with those proposed in~\cite{Ghosh_Universal,canonne2021discrete,sadeghi2020differentially,McSherry2007Mech}. In plot (a), the optimal noise mechanism is compared with  the discrete geometric mechanism, discrete Gaussian mechanism and Gumbel mechanism for a fixed MSE, i.e., $\rho^{MSE}_{Geo.} = \rho^{MSE}_{Gau.} = \rho^{MSE}_{Gum.} = \rho^{MSE}_{Opt.} = 0.6101$ and $n=8$.} In plot (b), the optimal noise mechanism is compared with the discrete geometric mechanism, discrete count mechanism, and data independent mechanism, i.e. $f(0) = 1-\rho^{ER}_{Ind.}$, $f(\eta) = \rho^{ER}_{Ind.}/n,~\eta \in [n]_+$, for a fixed ER, i.e., $\rho^{ER}_{Geo.} = \rho^{ER}_{Cnt.} = \rho^{ER}_{Opt.} = \rho^{ER}_{Ind.} = 0.3$ and $n=7$.}
     \label{Fig:Comparisons}
\end{figure}
In plot~\ref{fig:GeometricGaussianCompare}, we show the comparison of the proposed optimal noise mechanism with the discrete geometric, discrete Gaussian, and Gumbel mechanisms for a given MSE = 0.6101. 
Similarly, in plot~\ref{fig:GeometricCountCompare}, we show the comparison of the proposed optimal noise mechanism with the discrete Gaussian mechanism, discrete count mechanism, and data independent mechanisms for a given ER = 0.3. From the plots, it is clear that the proposed optimal noise mechanism significantly outperforms all these mechanisms $\forall \epsilon > 0$.
\begin{figure}
     \centering
     \includegraphics[width=0.45\textwidth]{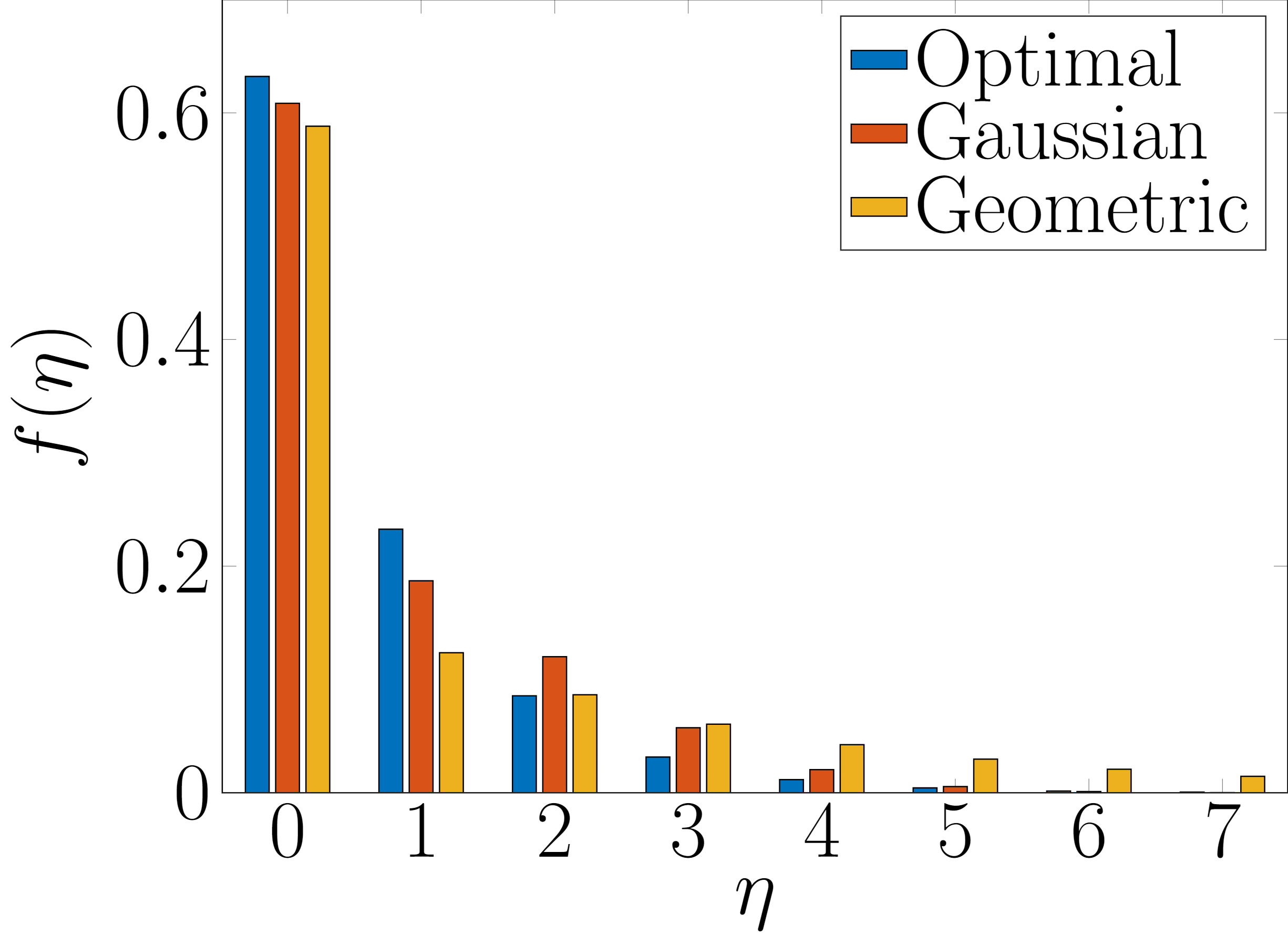}
     \caption{{This plot shows the PMF of the optimal noise mechanism compared with both the Gaussian and geometric mechanisms. The following parameter values are used: $n=8, \alpha = 0.7, \epsilon =1$, {$\delta = 0$}, and $\sigma^2 = 3.38$.}}
     \label{fig:PMF_comparison}
\end{figure}

\begin{figure}
     \centering
     \subfloat[][]{\includegraphics[width=0.23\textwidth]{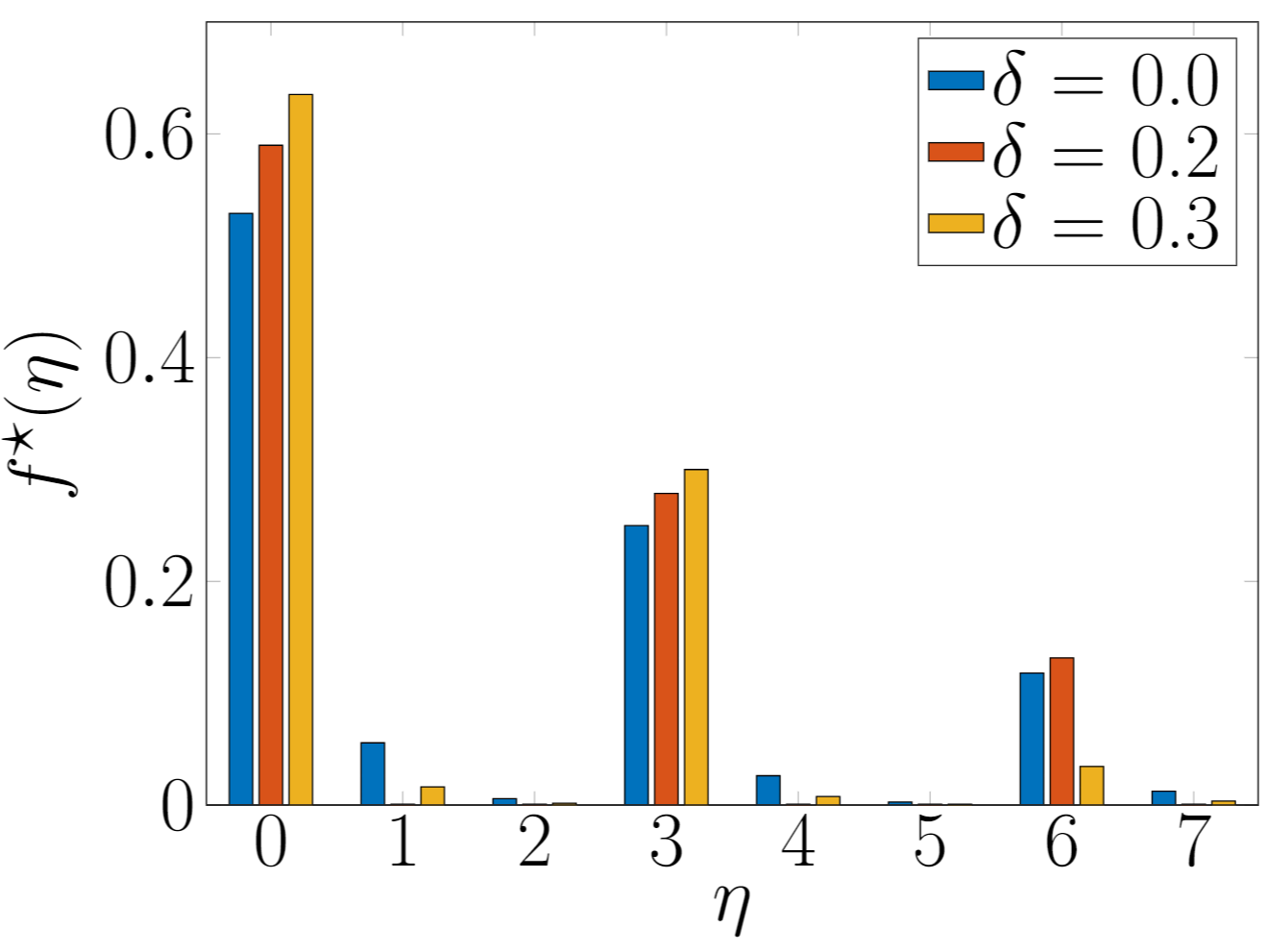}\label{fig:PMF_SD1}}
     \subfloat[][]{\includegraphics[width=0.23\textwidth]{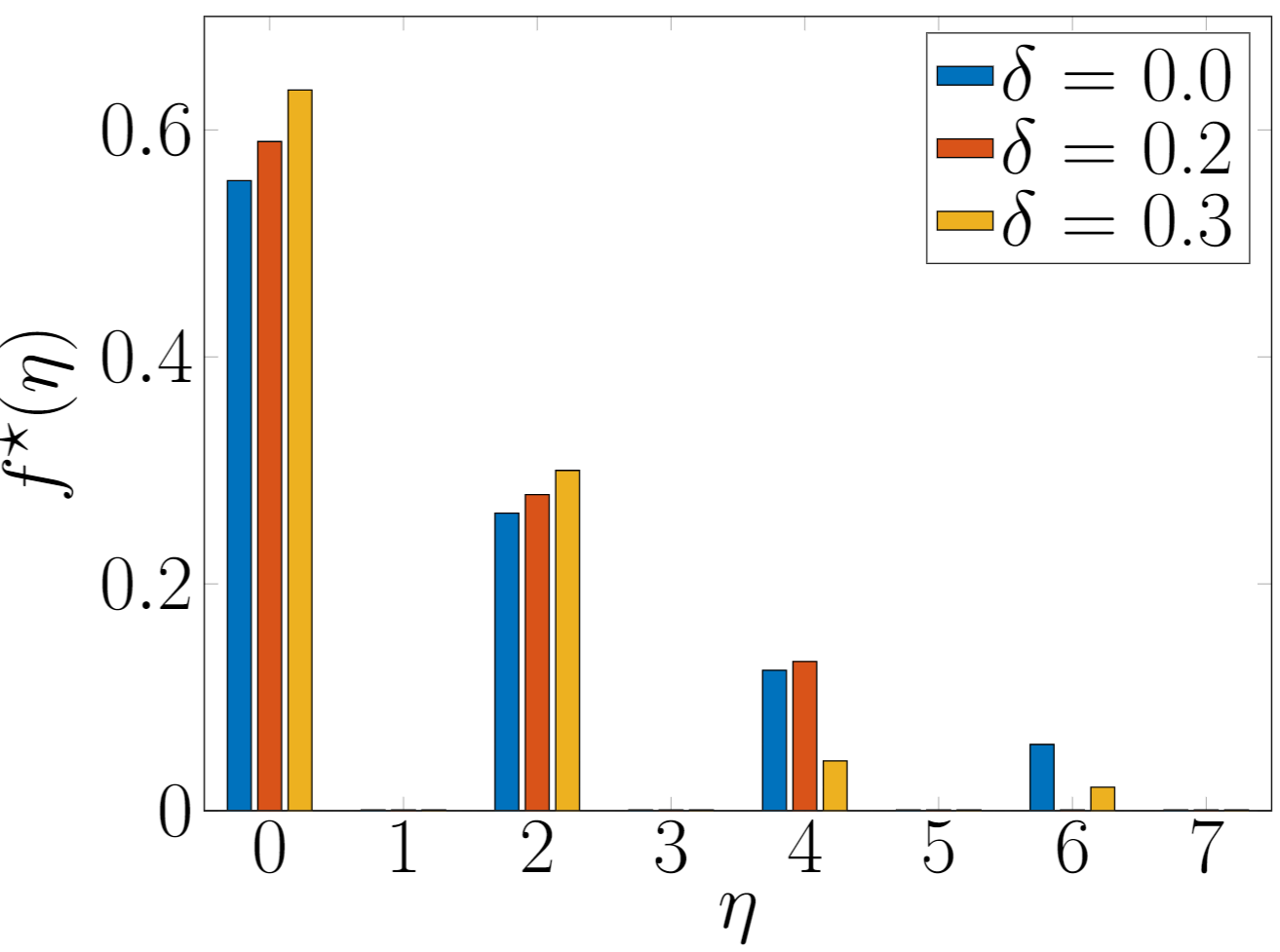}\label{fig:PMF_SD2}}
     \caption{These plots show the PMF of the optimal noise mechanisms for the SD neighborhood case and the following parameter values are used for both $n=7, \epsilon = 0.75$. $\hat{\mu} = 3$ is used in plot (a), whereas $\hat{\mu} = 2$ in used in plot (b).}
     \label{Fig:SD_PMF_plots}
\end{figure}

{
Now, we provide the comparison between the PMFs of optimal noise distribution with regard to Gaussian and geometric distributions in Fig.~\ref{fig:PMF_comparison} for the same MSE parameter for all the distributions. From the plot, we can observe that the probability mass at $\eta=0$ is maximum for the proposed mechanism, which validates our claim of the least error rate among these mechanisms.  
}

In the following figures, we show the structure of the PMF associated with the optimal noise mechanism. First, we consider the SD neighborhood case. 
More specifically, the plots in Fig.~\ref{Fig:SD_PMF_plots} show the PMF of the optimal noise mechanisms, $f^\star{(\eta)}, \eta \in [n]$, for SD neighborhood case for $\hat{\mu} = 3$ in Fig.~\ref{fig:PMF_SD1} and for $\hat{\mu} = 2$ in Fig.~\ref{fig:PMF_SD2}. In the left plot, we see that $f^\star{(\eta)}$ is non-zero for all $\eta \in [n]$ since ($n+1, \hat{\mu}$) are relatively prime but in the right plot, we see that $f^\star{(\eta)}$ is zero for $\eta \in \{1, 3, 5, 7\}$, since ($n+1, \hat{\mu}$) are not relatively prime (see Fig.~\ref{Fig:Periodic_Explain} for the reasoning). From the plots we observe that as $\delta$ increases, $f^\star{(0)}$ increases and since error rate, $\rho^{ER}$, is $1 - f^\star{(0)}$ it decreases with increase in $\delta$. And in the right plot, at some of the $\eta \in [n]_+$ values, zero probability mass is assigned. Hence, we see a higher value of $f^\star{(0)}$ compared to the corresponding values in the left plot, thus having a lower error rate in the right plot for a given $\delta$ value. Recall Lemma~\ref{Lemma:ordered_set} and Lemma~\ref{Lem:SD_delta0}, to see that $f^\star{(0)} \ge f^\star{(3)} \ge f^\star{(6)} \ge f^\star{(1)} \ge f^\star{(4)} \ge f^\star{(7)} \ge f^\star{(2)} \ge f^\star{(5)}$ in Fig.~\ref{fig:PMF_SD1} which are represented using $f^\star_{(i)}, i \in [n]$, respectively. Similarly, we observe  $f^\star{(0)} \ge f^\star{(2)} \ge f^\star{(4)} \ge f^\star{(6)}$ and rest $f^\star{(1)} = f^\star{(3)} = f^\star{(5)} = f^\star{(7)} = 0$ in Fig.~\ref{fig:PMF_SD2}. 
\begin{figure}
     \centering
     \subfloat[][]{\includegraphics[width=0.23\textwidth]{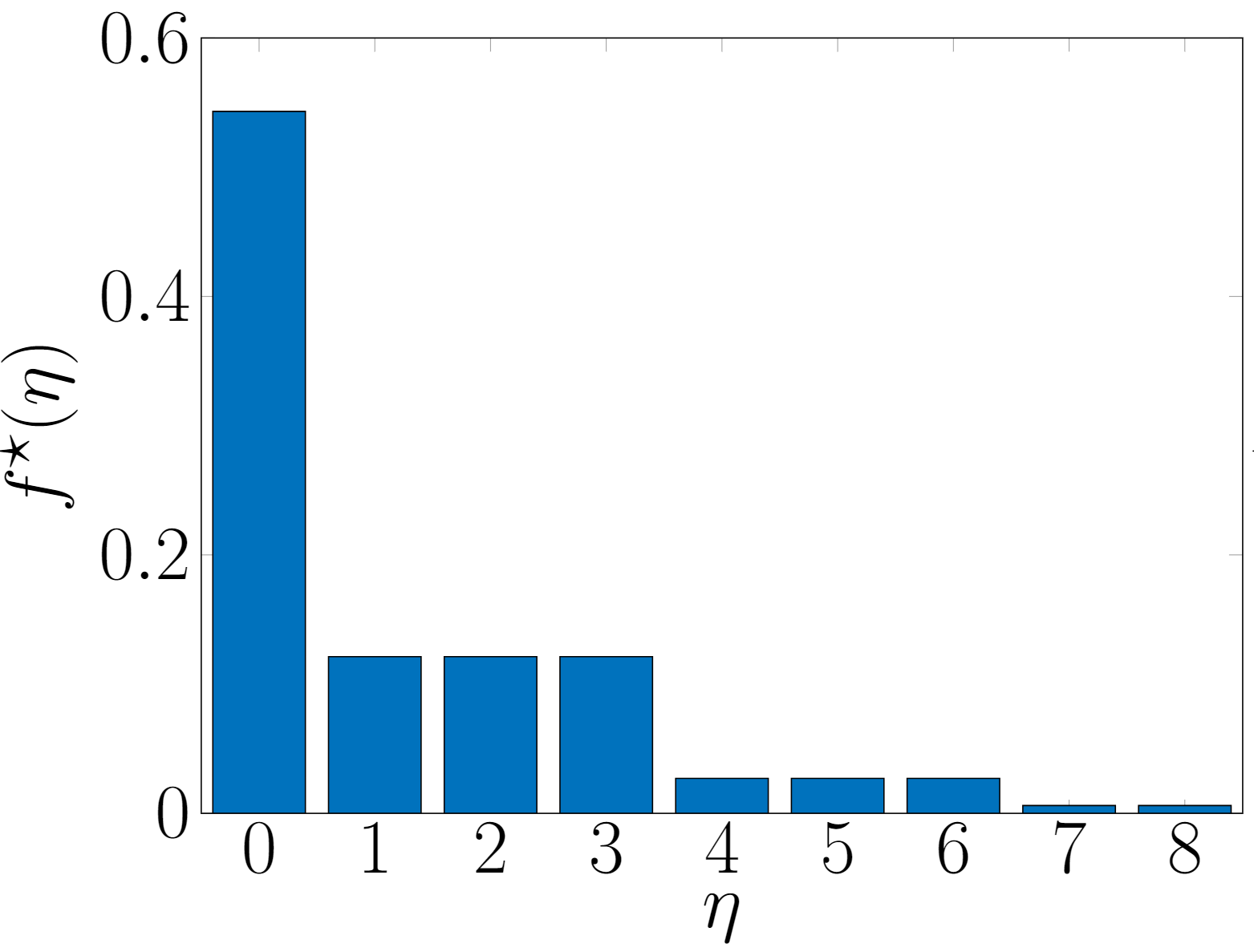}\label{fig:PMF_BD1}}
     \subfloat[][]{\includegraphics[width=0.23\textwidth]{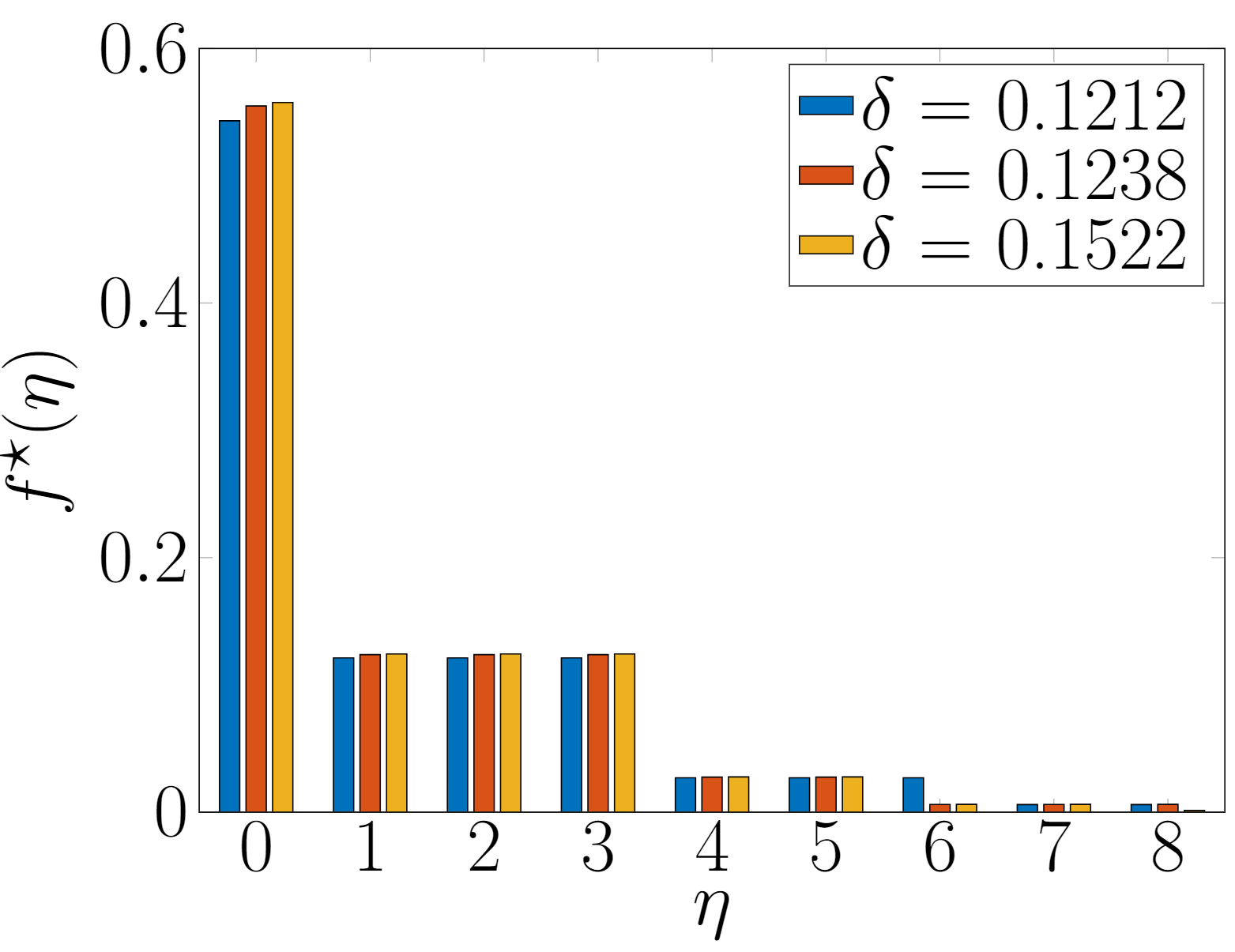}\label{fig:PMF_BD2}}
     \caption{These plots show the PMF of the optimal noise mechanisms for BD neighborhood case and the following parameter values are used: $n=8, \bar{\mu} = 3, \epsilon = 1.5$.  In plot (a) we see that the optimal noise mechanism follows staircase pattern starting from $\eta = 1$ with $b=2$ steps of length $\bar{\mu} = 3$ and one last step of length $r=2$. In plot (b) we show how staircase pattern and step lengths change with $\delta$. It can be seen at $\delta = \overline{\delta}^{\epsilon}_{0} = 0.1212$, $\delta = \underline{\delta}^{\epsilon}_{1} = 0.1238$, and $\delta = \underline{\delta}^{\epsilon}_{2} = 0.1522$  step lengths are: (1,3,3,2) (blue coloured bars), (1,3,2,3) (red coloured bars), and (1,3,2,2,1) (yellow coloured bars), respectively.}
     \label{Fig:BD_PMF_plots}
\end{figure}
Next, we consider the BD neighborhood case. The plots in Fig.~\ref{Fig:BD_PMF_plots} show the PMF of the optimal noise mechanisms, $f^\star{(\eta)}, \eta \in [n]$, for BD neighborhood case for $\delta = 0$ in Fig.~\ref{fig:PMF_BD1} and for $\delta > 0$ in Fig.~\ref{fig:PMF_BD2}. In the left plot, we clearly see the staircase pattern with step sizes equal to $\bar{\mu}$, except for the last step. In the right plot, we see that step widths change as $\delta$ increases its values while the staircase structure is maintained. 
Note that the vertical height of each step is $e^{\epsilon}$ times higher than the previous one, as can be seen in Fig.~\ref{Fig:BD_PMF_plots} and in Table~\ref{Table:BD_Case}. Also, from the plot in Fig.~\ref{fig:PMF_BD2} and Table~\ref{Table:BD_Case}, one can observe as the value of $f^\star{(0)}$ increases (thus the value of $\rho^{ER}$ decreases) as $\delta$ increases, as it is expected. 
\begin{table}[ht]
\centering
\caption{The PMF of the optimal noise mechanism for different values of $\delta$ for $n=8, \bar{\mu} = 3, \epsilon = 1.5$.}
\scalebox{0.9}{
\begin{tabular}{ |c|c|c|c|c| } 
 \hline
    &  $\delta=0$ & $\delta=0.1212$ & $\delta=0.1238$ & $\delta=0.1522$\\
  \hline
 $f^\star{(0)}$ & 0.5432 & 0.5432 & 0.5548 & 0.5575 \\ 
 $f^\star{(1)}$ & 0.1212 & 0.1212 & 0.1238 & 0.1244 \\ 
 $f^\star{(2)}$ & 0.1212 & 0.1212 & 0.1238 & 0.1244 \\
 $f^\star{(3)}$ & 0.1212 & 0.1212 & 0.1238 & 0.1244 \\
 $f^\star{(4)}$ & 0.0270 & 0.0270 & 0.0276 & 0.0278 \\
 $f^\star{(5)}$ & 0.0270 & 0.0270 & 0.0276 & 0.0278 \\
 $f^\star{(6)}$ & 0.0270 & 0.0270 & 0.0062 & 0.0062 \\
 $f^\star{(7)}$ & 0.0060 & 0.0060 & 0.0062 & 0.0062 \\
 $f^\star{(8)}$ & 0.0060 & 0.0060 & 0.0062 & 0.0014 \\
 \hline
\end{tabular}}
\label{Table:BD_Case}
\end{table}

Next, we consider the vector query case and provide a counter example to support Remark~\ref{Rem:VectorQuery} for a two-dimensional vector query. Let $n=4, \epsilon_1 =1.5, \epsilon_2 =1.5, \epsilon = 3, \mu_{XX'} = \{0,1,2\}, \delta=0$. The optimal noise mechanism for  $(\epsilon_1,0)$ and $(\epsilon_2,0)$ DP are: $f^{\star}_1 (\eta)$ = $f^{\star}_2 (\eta)$ = $[0.6469,  0.1443, 0.1443, 0.0322, 0.0322]$; the values of the optimal noise mechanism PMF $f^{\star} (\eta_1, \eta_2)$ for  $(\epsilon,0)$ are in \eqref{eq:jointPMFvalues}:
\begin{equation}\label{eq:jointPMFvalues}
\begin{bmatrix}
 0.6954 &   0.0346 &   0.0346 &   0.0017 &   0.0017\\
  0.0346 &   0.0346 &   0.0346  &  0.0017 &   0.0017 \\
  0.0346 &   0.0346 &   0.0346  &  0.0017 &   0.0017 \\
  0.0017 &   0.0017 &   0.0017  &  0.0017 &   0.0017 \\
  0.0017 &   0.0017 &   0.0017  &  0.0017 &   0.0017 \\
\end{bmatrix}
\end{equation}

The marginal distributions happen to be equal, which makes sense in terms of symmetry: $f_1 (\eta_1) = \sum_{\eta_2=0}^n f^{\star} (\eta_1, \eta_2)\equiv f_2 (\eta_2)$ and they have masses  [0.7681, 0.1073,   0.1073,   0.0086, 0.0086]. We can observe that $f^{\star} (\eta_1, \eta_2) \ne f_1 (\eta_1)f_2 (\eta_2)$.

The plots in Figs.~\ref{Fig:ErrorRate_mu_plots}--\ref{Fig:PMF_AMI_plots} are self explanatory. 

\begin{figure}[!ht]
\centering
\begin{subfigure}{.23\textwidth}
\centering
\begin{adjustbox}{width = 1\columnwidth}
\includegraphics[width=0.99\textwidth]{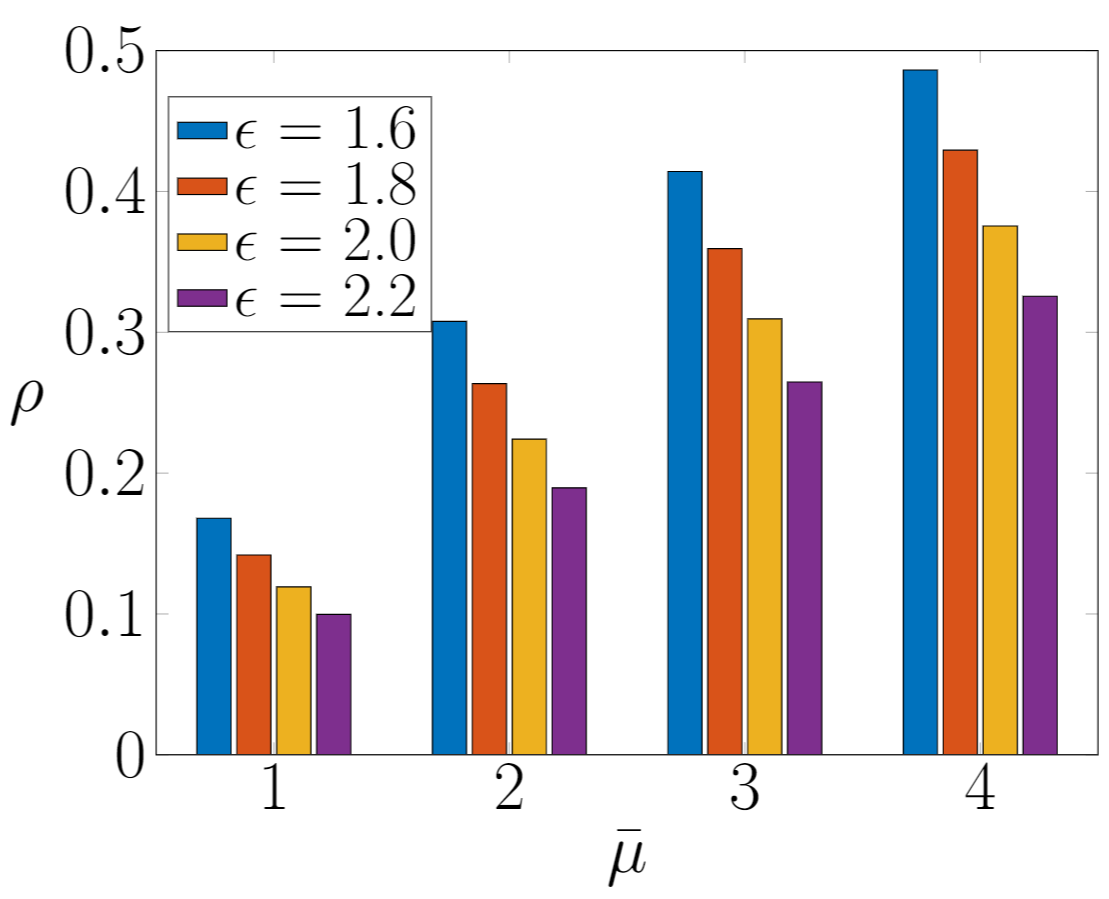}
\end{adjustbox}
\label{fig:ErrorRate_mu_eps}
\end{subfigure}%
\begin{subfigure}{.23\textwidth}
\centering
\begin{adjustbox}{width = 1\columnwidth}
\includegraphics[width=0.99\textwidth]{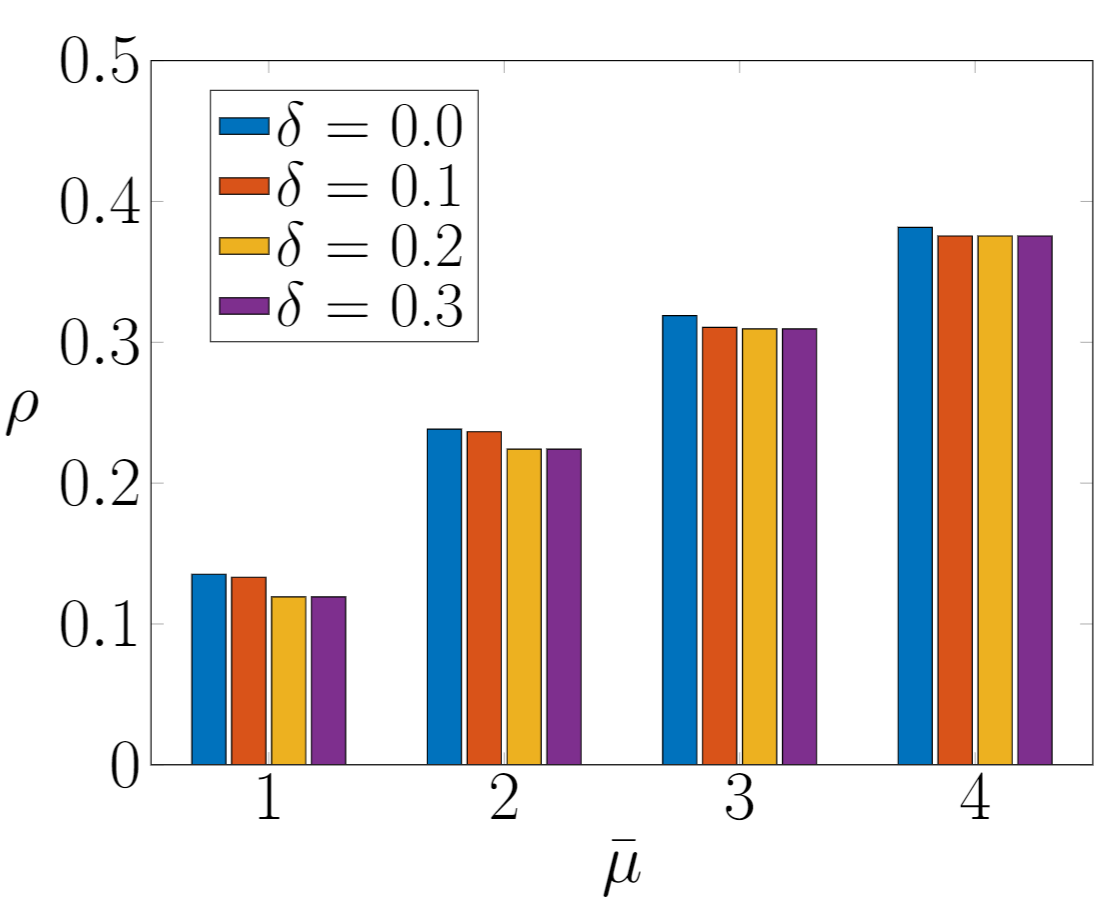}
\end{adjustbox}
\label{fig:ErrorRate_mu_delta}
\end{subfigure}
\caption{These plots show the error rate for the BD neighborhood case v/s parameter $\bar{\mu}$. In the left plot $n=9, ~\delta=0.2$, and in the right plot  $n=9, ~\epsilon=2$ are used.}
\label{Fig:ErrorRate_mu_plots}
\end{figure}

\begin{figure}[!ht]
\centering
\begin{subfigure}{.23\textwidth}
\centering
\begin{adjustbox}{width = 1\columnwidth}
\includegraphics[width=0.99\textwidth]{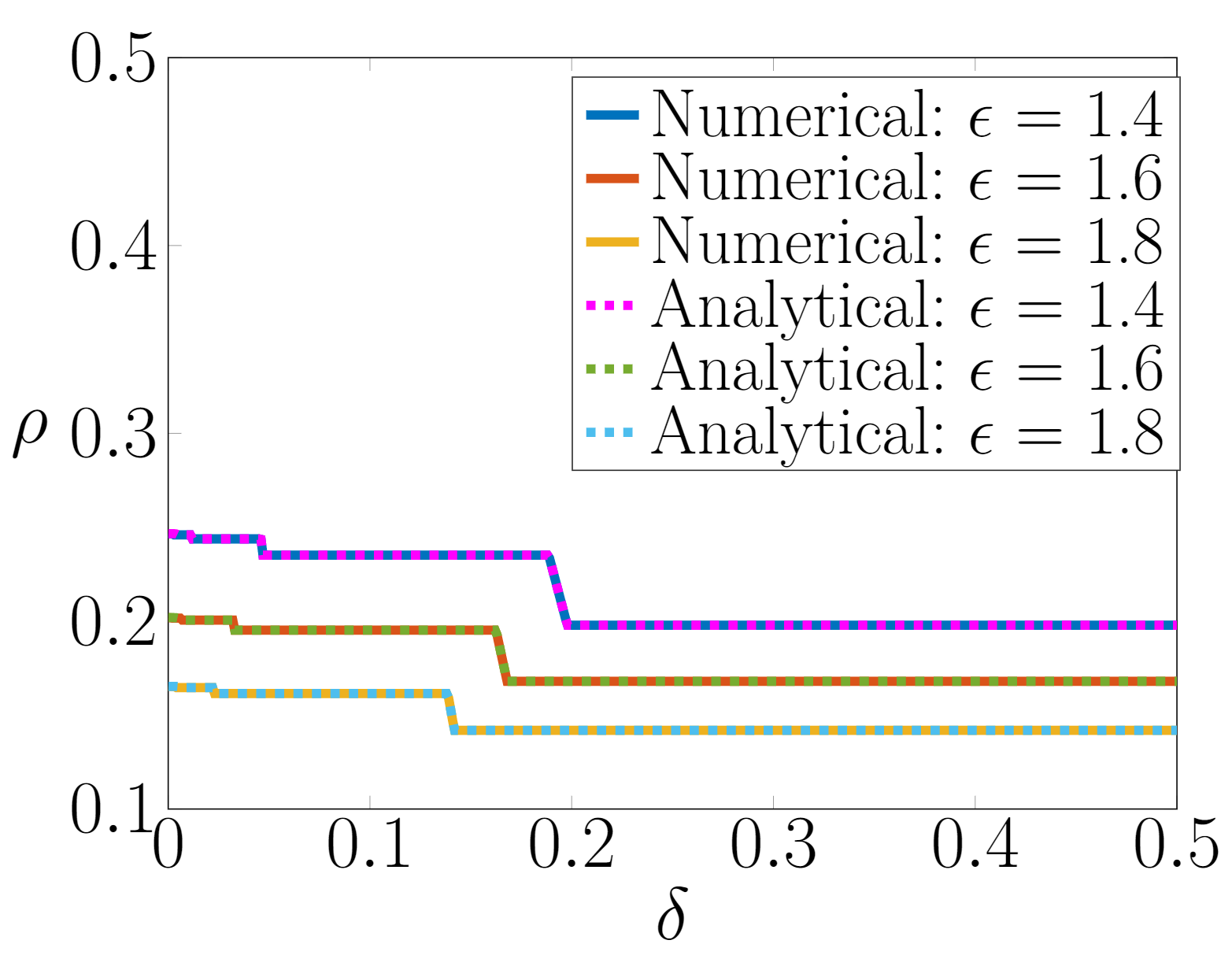}
\end{adjustbox}
\caption{}
\label{fig:ErrorRate_delta_SD}
\end{subfigure}%
\begin{subfigure}{.23\textwidth}
\centering
\begin{adjustbox}{width = 1\columnwidth}
\includegraphics[width=0.99\textwidth]{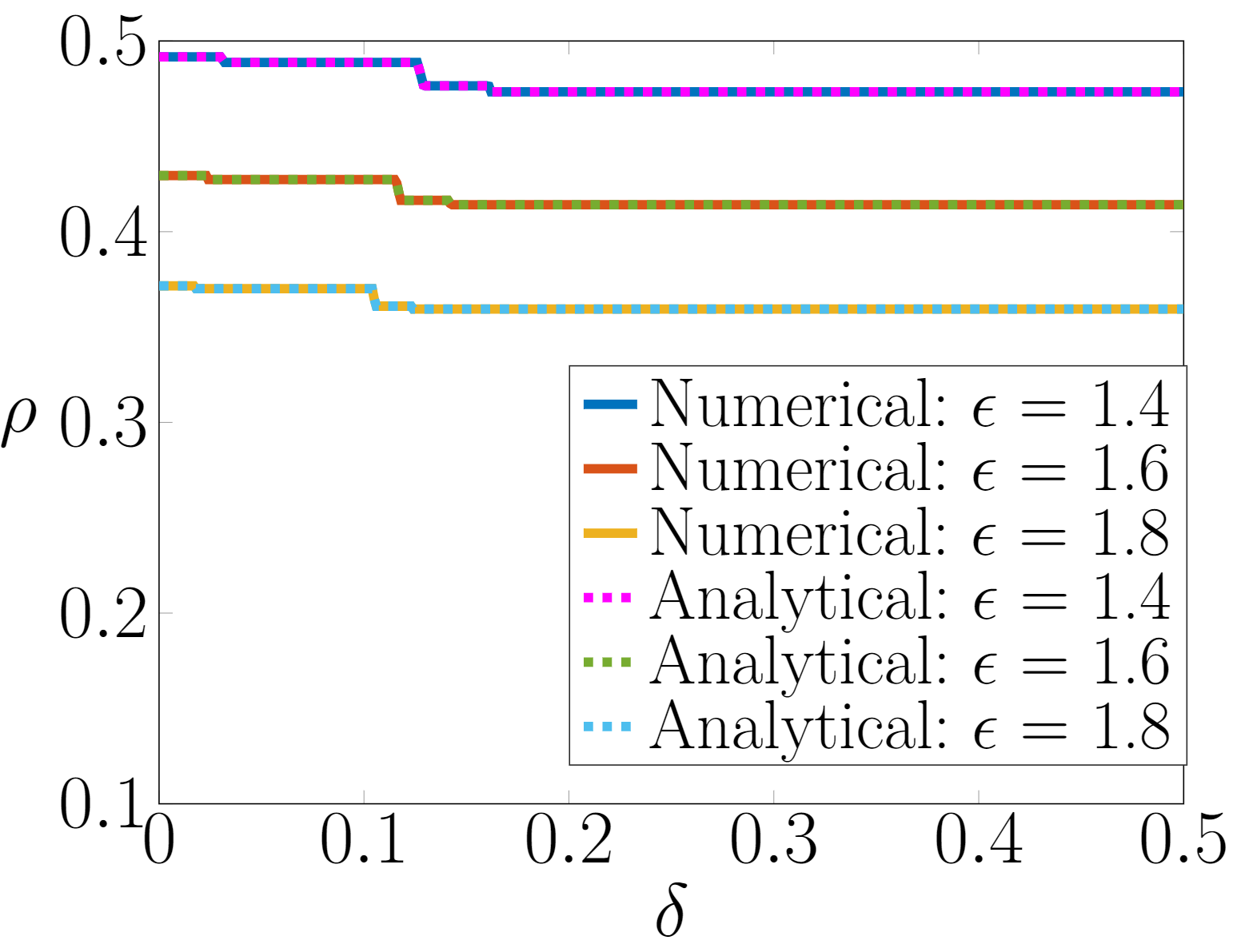}
\end{adjustbox}
\caption{}
\label{fig:ErrorRate_delta_BD}
\end{subfigure}
\caption{Error rate $\rho$ for the SD and BD neighborhood cases v/s parameter $\delta$, respectively, confirming the trends predicted by Theorems \ref{Thm:Linear_Flat_Region} and \ref{Thm:BD_results}. In plot (a) $n=9, ~\hat{\mu}=3$ and in plot (b) $n=9, ~\bar{\mu}=3$ are used.}
\label{Fig:ErrorRate_delta_plots}
\end{figure}

\begin{figure}[!ht]
\centering
\begin{subfigure}{.23\textwidth}
\centering
\begin{adjustbox}{width = 1\columnwidth}
\includegraphics[width=0.99\textwidth]{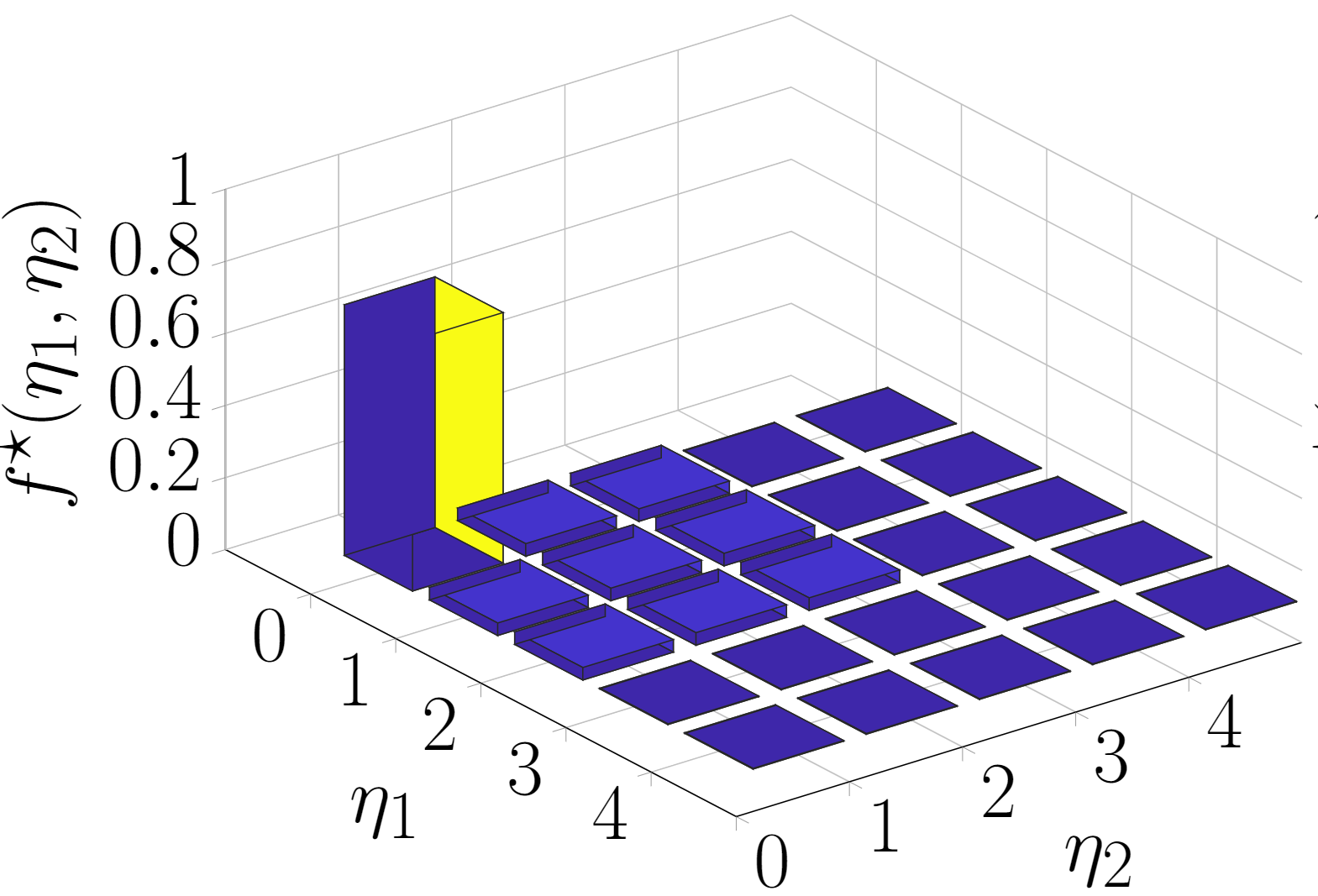}
\end{adjustbox}
\caption{}
\label{fig:PMF_2D_BD}
\end{subfigure}%
\begin{subfigure}{.23\textwidth}
\centering
\begin{adjustbox}{width = 1\columnwidth}
\includegraphics[width=0.99\textwidth]{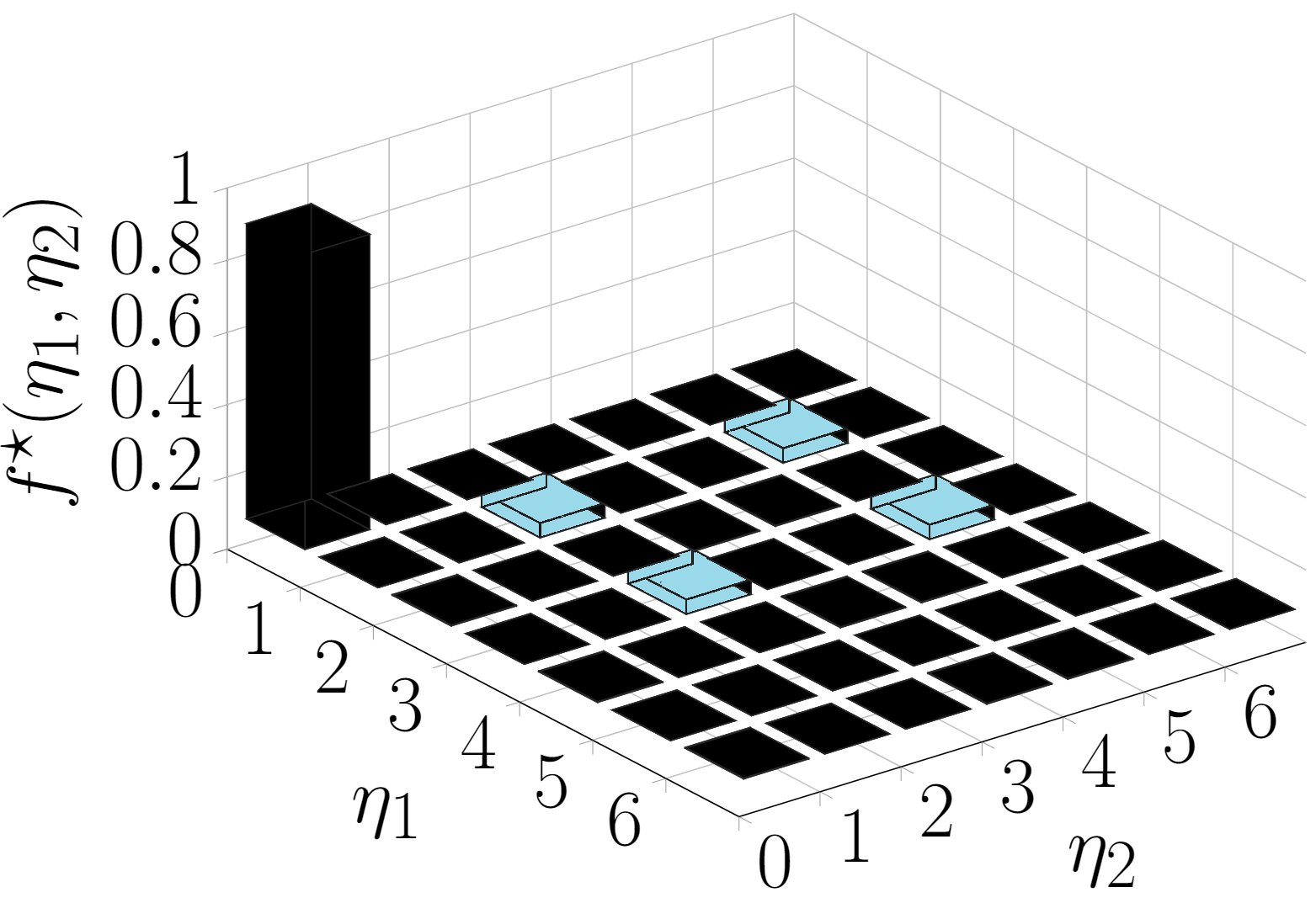}
\end{adjustbox}
\caption{}
\label{fig:PMF_2D_Arb}
\end{subfigure}
\caption{The optimal noise joint PMF for two-dimensional vector query case for a BD and arbitrary neighborhoods, respectively.  In plot (a) the parameters are: $n=4, ~\epsilon=3,  ~\bar{\mu}_1=\bar{\mu}_2=2$. In this BD neighborhood example, the staircase pattern is similar to the scalar query case in Theorem \ref{Thm:BD_results}. In  plot (b) the following parameters are used: $n=6, ~\epsilon=3,  ~{\mu}_1= \{1, 3\}, ~{\mu}_2= \{2, 5\}$. In this arbitrary  neighborhood example, the second largest peaks can be observed at the union of distance one set of each dimension, i.e., {$[1,2], ~[1,5],~[3,2], ~[3,5]$}.}
\label{Fig:PMF_2D_plots}
\end{figure}

\begin{figure}[!ht]
\centering
\begin{subfigure}{.23\textwidth}
\centering
\begin{adjustbox}{width = 1\columnwidth}
\includegraphics[width=0.99\textwidth]{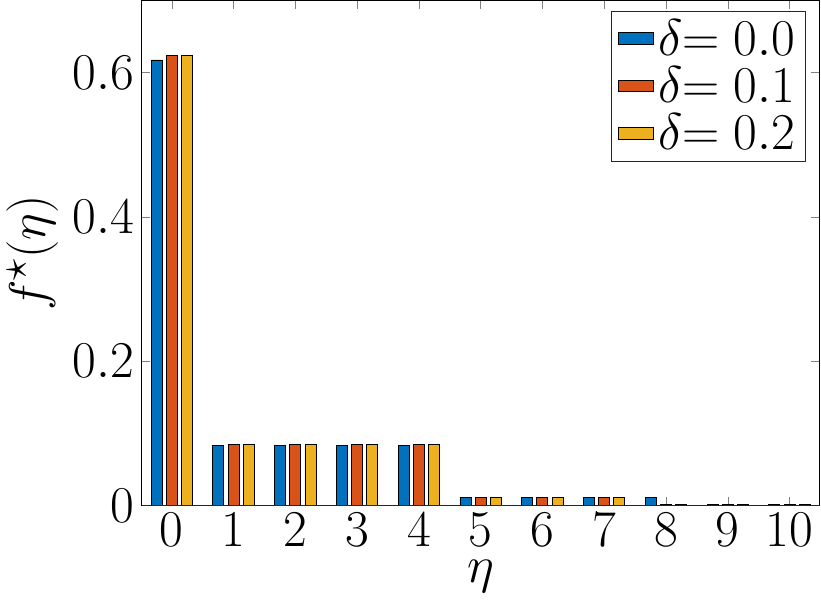}
\end{adjustbox}
\label{fig:PMF_AMI1}
\end{subfigure}%
\begin{subfigure}{.23\textwidth}
\centering
\begin{adjustbox}{width = 1\columnwidth}
\includegraphics[width=0.99\textwidth]{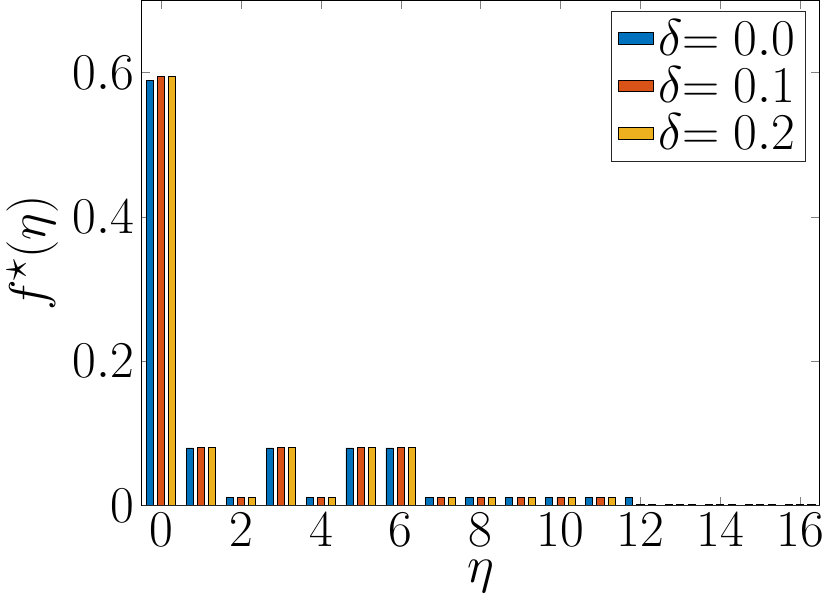}
\end{adjustbox}
\label{fig:PMF_AMI2}
\end{subfigure}
\caption{{{These plots show the PMF of the optimal noise mechanisms when the Advanced Metering Infrastructure (AMI) database is queried from 1416 houses that belong to 12 distribution circuits across California, USA. We use $\tilde{\mu} = \bigcup_{\forall X \in {\mathcal X}} \bigcup_{\forall X'\in {\mathcal X}_{X}^{(1)}} \mu_{\XX}$. In the left plot, 11 quantization levels are used, hence $n=10$. In this example, $\tilde{\mu} = \{1, 2, 3, 4\}$ is observed. In the right plot, 17 quantization levels are used, hence $n=16$. In this example, $\tilde{\mu} = \{1, 3, 5,6\}$ is observed. From these figures we observe that $f^\star(\eta)$ is considerably larger for $\eta \in \{\tilde{\mu} \cup 0\}$ than those $\eta \in [n] \setminus \{\tilde{\mu} \cup 0\}$.}}}
\label{Fig:PMF_AMI_plots}
\end{figure}

\section{Conclusions and Future Work}
Considering queries whose domain is discrete and finite,  in this paper we proposed a novel MILP formulation to determine what is the PMF for an additive noise mechanism that minimizes the error rate of the DP answer for any ($\epsilon,~\delta$) pair. The modulo addition between the noise and the queried data is modulo $n+1$ equal to the size of the query domain. 
For two special cases, which we referred to as the SD neighborhood and bounded difference (BD) neighborhood, we have provided closed-form solutions for the 
optimal noise PMF and its probability of error versus $\delta$ for a given $\epsilon$ and studied the asymptotic case for $n\rightarrow \infty$.
We also compared the proposed optimal noise mechanism to state-of-the-art noise mechanisms and found that it significantly outperforms them for a given ER or MSE. 
In the future, we plan to leverage these results in several applications that have to do with labeling data as well as a building block to study theoretically queries with finite uncountable support as well as the case of vector queries, whose optimum PMF can be calculated with our MILP and does not appear to be the product of the optimum PMFs for each entry. 

\section*{Appendix}
\section{Proof of Proposition~\ref{prop:PostProcessing}}\label{Apdx:prop1Proof}
{
Let us define the domains of $q, \qtilde,$ and $g(\qtilde)$ to be ${\mathcal Q}_1, {\mathcal Q}_2$, and ${\mathcal Q}_3$, respectively. The given function is $g: {\mathcal Q}_2 \rightarrow {\mathcal Q}_3$. We can prove that $(\epsilon,\delta)-\text{PDP}$ is preserved in general under post-processing if $g(\qtilde)$ is bijective.
In fact, in this case ${\mathcal Q}_2={\mathcal Q}_3$ and the probability mass of $g(\qtilde)$ in ${\mathcal Q}_3$ domain, i.e., $f_{g(\qtilde)}(g(\qtilde))$, is equal to the probability mass of the corresponding argument $\qtilde$. So, $\forall X\in {\mathcal X}, \forall {X'\in {\mathcal X}_{X}^{(1)}}$, it is trivial to see that:
\begin{equation}
\label{eq:post_proc_1}
    L_{\XX} (g(\qtilde)) = L_{\XX} (\qtilde) 
\end{equation}
and that the probability $\delta$ of the leakage event $L_{\XX} (g(\qtilde))>\epsilon$ remains the same. 
If $\delta=0$ then the $\epsilon-PDP$ is always preserved. The case $|{\mathcal Q}_2|\geq |{\mathcal Q}_3|$ is the interesting one (e.g. clamping): in this case multiple values of $\qtilde \in {\mathcal V}_g$ map onto a single value $g(\qtilde)=g$. Next, we show that $L_{\XX} (\qtilde)\leq \epsilon$ with probability one implies $L_{\XX} (g(\qtilde))\leq \epsilon$ with probability one. In fact, since for all $\qtilde$, $f(\qtilde|X)\leq e^{\epsilon}f(\qtilde|X')$:
\begin{align}
    L_{\XX} (g(\qtilde))&=\log
    \frac{f(g|X)}{f(g|X')}=
    \log
    \frac{\sum_{\qtilde\in {\mathcal V}_g}f(\qtilde|X)}{
    \sum_{\qtilde\in {\mathcal V}_g}f(\qtilde|X')}\notag\\
    &\leq 
    \log
    \frac{\sum_{\qtilde\in {\mathcal V}_g}e^{\epsilon}f(\qtilde|X')}{
    \sum_{\qtilde\in {\mathcal V}_g}f(\qtilde|X')}\leq \epsilon
\end{align}
}

\section{Proof of Lemma~\ref{Lemma:ordered_set}}\label{Apdx:Lem1Proof}
When the objective is minimizing error rate, it is natural to introduce the inequality constraint  $f^\star_{(0)} = f^\star(0) \ge 0.5=1-\rho^{ER}$, 
. Hence, $f^\star_{(0)} = \sup_{\forall k \in [n]} f^\star_{(k)}$. Now, substitute $\eta=0$ and $\mu_{\XX} = \hat{\mu}$ in~\eqref{Eq:f0_constraint} we see that $f^\star(0) \le e^{\epsilon}f^\star(\hat{\mu})$. 
Since we are minimizing the sum of the mass away from zero, we assign $f^\star(\hat{\mu})$ to the minimum possible value, which is $f^\star(\hat{\mu}) = e^{-\epsilon}f^\star(0)$ in this case.  Similarly, from~\eqref{Eq:f0_constraint} we see that $f^\star(k\hat{\mu}) \le e^{\epsilon}f^\star((k+1)\hat{\mu})$, $k \in [n-1]_+$, and we are minimizing the sum of the mass away from zero, we need to assign $f^\star((k+1)\hat{\mu})$ to the minimum possible value, which is $f^\star((k+1)\hat{\mu}) = e^{-\epsilon}f^\star(k\hat{\mu})$.\footnote{Note that for $f^\star(n\hat{\mu})$, there is no choice to assign any value since it will be automatically fixed once $f^\star(k\hat{\mu})$, $k \in [n-1]$ values are fixed and it is trivial to see that $f^\star(n\hat{\mu}) = \min f^\star(k\hat{\mu}), \forall k \in [n]$.} Since $\epsilon > 0$, $f^\star(k\hat{\mu}) \ge f^\star((k+1)\hat{\mu})$, $\forall k \in [n]$. Hence, we can write $f^\star_{(k)} = f^\star(k\hat{\mu})$, $\forall k \in [n]$.
\section{Proof of Lemma~\ref{Lem:SD_delta0}}\label{Apdx:Cor1Proof}
\textbf{Case 1}: $(\hat{\mu}, (n+1))$ are relatively prime.\\
The proof logic is as follows. Since minimizing the error rate is equivalent to having the maximum mass possible at $\eta=0$, we expect $f^\star(0) = \sup_{\eta \in [n]}f^\star(\eta)$. The constraint~\eqref{Eq:f_eta_ordered}, implies $f^\star(\hat{\mu}) \ge e^{-\epsilon} f^\star(0)$. Also that for any $\eta = k\hat{\mu},~k\in [2:n]$, $f^\star(k\hat{\mu}) \ge e^{-\epsilon} f^\star((k-1)\hat{\mu}) \ge e^{-k\epsilon} f^\star(0)$ and $f^\star((n+1)\hat{\mu}) \ge e^{-\epsilon} f^\star(n\hat{\mu})$. From all these inequalities and the constraint~\eqref{Eq:Opt_prob}, we conclude that what would allow having the largest mass of probability at $\eta=0$ is meeting all constraints as equality, starting from the first. Since $\hat{\mu}$ and $n$ are prime, the multiples of $\hat{\mu}$ eventually cover the entire range $[n]$ and therefore:  $f^\star(k\hat{\mu}) = e^{-k\epsilon} f^\star(0),~k \in [n]$. This result leads to the optimum distribution in  Lemma~\ref{Lem:SD_delta0}.  Now, $f^\star(0)$ can be computed as follows:
\begin{subequations}
\begin{align}
    f^\star(0) &+ e^{-\epsilon}f^\star(0) + \ldots + e^{-n\epsilon}f^\star(0) = 1 \\
    \Rightarrow\quad f^\star(0) & = \frac{1 - e^{-\epsilon}}{1 - e^{-(n+1)\epsilon}}.
\end{align}
\end{subequations}

\noindent \textbf{Case 2}: $(\hat{\mu}, (n+1))$ are not relatively prime.\\
The same argument of Case 1 holds for $k \in [N_{\hat{\mu}}-1]_+$, where $N_{\hat{\mu}} = \frac{(n+1)}{\gcd((n+1),\hat{\mu})}$,
i.e., $f^\star(k \hat{\mu}) =  e^{-k \epsilon} f^\star(0), ~~~~k \in [N_{\hat{\mu}}-1]_+$. However in this case, since for $k = N_{\hat{\mu}}$, $f^\star(N_{\hat{\mu}} \hat{\mu}) = f^\star(0)$ and the cycle repeats over the same exact values covered from zero up to $(N_{\hat{\mu}}-1)$ which does not include all PMF entries. Since we are minimizing the objective function to satisfy all the inequality constraint $f^\star_{(0)} = f^\star(0) \ge 0.5=1-\rho^{ER}$  and~\eqref{Eq:f_eta_ordered} for all $k$ values that are not constraining $f^\star(0)$, the best choice is to assign them zero, i.e., $f^\star(k) = 0,~k \in [n] \setminus [i\hat{\mu}], ~\forall i \in [N_{\hat{\mu}}-1]$.

\section{Proof of Theorem~\ref{Thm:Linear_Flat_Region}}\label{Apdx:Thm3Proof}
We focus on Case 1, as Case 2 follows from Remark \ref{Rem:PrimeRemark}. In Lemma \ref{Lemma:ordered_set}, we clarified that it is best to deal with ordered values, and in Lemma~\ref{Lem:SD_delta0}, we specified $f^\star_{(h)}$, $h \in [n]$, as a function of $\epsilon >0$ for  $\delta=0$. The best solution for $\rho(\epsilon,\delta)$ initially does not change until violating an inequality in \eqref{Eq:f_eta_ordered} yields better accuracy. This happens as soon as the second smallest value $f^{\star}_{(n-1)}$ corresponding to $\delta=0$ is $f^{\star}_{(n-1)}=\delta$, which is the upper-limit $\overline{\delta}_0$ from  \eqref{eq:delta-top-bottom}. The reason why it is $f^{\star}_{(n-1)}$ and not $f^{\star}_{(n)}$ matters because surely $f^{\star}_{(n)}<e^{\epsilon}f^{\star}(0)$ which in the modulo $n$ sum is the value that follows and that we aim at maximizing. At this point, for $\overline{\delta}_0<\delta\leq \underline{\delta}_1$ the value of  $f^{\star}_{(n-1)}=\delta$, all the values for $0\leq h<n-1$ meet the constraints with equality and thus $f^{\star}_{(h)}=e^{(n-1-h)}\delta$, while the last value $f^{\star}_{(n)}$ progressively diminishes until it becomes zero, as shown in equation \eqref{eq:SD_f*(h)linear}, at the start of the next flat region. This pattern continues until eventually one by one all $n-1$ masses become zero except for $f^{\star}(0)=1=\underline{\delta}_n^{\epsilon}$.   
\section{Proof of Theorem~\ref{Thm:BD_results}} \label{Apdx:Thm4Proof}
In Lemma~\ref{Lem:BD_delta0}, we have the expressions for $\phi_i, ~i \in [b+1]_+$ for $\epsilon>0$ and $\delta=0$. The best solution for 
$\rho^{\star}(\delta,\epsilon)$ does not change w.r.t $\delta>0$ until violating an inequality in \eqref{Eq:f_eta_ordered_BD} to yield better accuracy. The key to the proof is understanding that the first inequality to be violated occurs when considering $\delta$ greater or equal not to the smallest PMF value but to the PMF values of the third to the last group in $\delta=0$ case, i.e., $\phi_{b-1} = \delta$, whose value equals the first boundary point $\overline{\delta}^{\epsilon}_{0,0}$ (see~\eqref{eq:delta_overline_BD}). The inequality violated is with respect to the PMF of the second to the last group, which is $\phi_{b}$, which becomes $=\delta$ and $\phi_{b-1}>e^{\epsilon}\delta$.  The reason why the PMF of the last group, i.e. $\phi_{b+1}$, does not violate the inequality in \eqref{Eq:f_eta_ordered_BD} is that $\phi_{b+1} < e^{\epsilon}\phi_0$ (we use $\phi_0$ since we are doing modulo summation) is always true for all members of this group, due to the fact that $\phi_0$ is the objective function which we are maximizing.  Similarly, the PMF of the second last group, i.e. $\phi_{b}$, does not violate the inequality in \eqref{Eq:f_eta_ordered_BD} because some members of this group do not violate the inequality $\phi_{b} < e^{\epsilon}\phi_0$ for the same reason as stated above. At this point, for $\overline{\delta}_{0,0}<\delta\leq \underline{\delta}_{1,1}$ the value of  $\phi_{b-1}=\delta$, the PMF values for $i \in [b-1]$ meet the constraints with equality and thus $\psi_i^1(\delta)=\delta e^{(b-1-i)\epsilon}$, for $i \in [b-1]$. At this point, the group with PMF $\phi_b$, whose length is $\ell_b=\bar{\mu}$, splits into two groups, one of length $\ell_b=(\bar{\mu}-1)$ and the other with a singleton step $\ell_{b+1}=1$. This split happens to assign more probability mass at $\psi_0^1(\delta)$, which is our objective, while still violating the constraint in \eqref{Eq:f_eta_ordered_BD} between $\delta = \psi_{b-1}^1(\delta)$ and $\psi_{b+1}^1(\delta)$ in order to lower the error further. For $i=b$ and $i=b+2$, the PMF values satisfy the constraints with equality, hence $\psi_{b}^1(\delta)=\delta e^{-b\epsilon}$ and  $\psi_{b+2}^1(\delta)=\delta e^{-(b+1)\epsilon}$ while the unconstrained singleton step PMF $\psi_{b+1}^1(\delta)$ decreases until it joins the next group since it has matched its value, and becomes $\phi^1_{b+1}$, as shown in equation \eqref{eq:PMF_BD_thm}, at the start of the next flat region, i.e., for  $\underline{\delta}_{1,1} \leq \delta\leq \overline{\delta}_{1,1}$. In this region, the PMF of every group is $e^{-\epsilon}$ times the PMF of previous group similar to  Lemma~\ref{Lem:BD_delta0}; the only change here is the lengths of $b^{th}$ and $(b+1)^{th}$ groups which are now $\ell_b = (\bar{\mu}-1)$ and $\ell_{b+1} = r+1$. 

Now we provide the reason for splitting only groups of lengths  $\bar{\mu}$ using contradiction. Suppose there is a group $i$ of length $\ell < \bar{\mu}$ which splits at the beginning of a linear region $k \in [b]_+$. As we know, for this group to split there must be a violation of the inequality in~\eqref{Eq:f_eta_ordered_BD} between the PMF values $\phi_{i-1}^{k-1} = \delta$ and $\phi_i^{k-1}$. Now, $\phi_i^{k-1}$ splits into  $\psi_{i}^k(\delta)$ of length $(\ell-1)$ and $\psi_{i+1}^k(\delta)$ of length 1, which decreases as $\delta$ grows. We see that $\psi_{i+1}^k(\delta)$ is $\leq \bar{\mu}$ distance away from at least two members of the $\psi_{i-1}^k(\delta)$ group leading to at least two violations in the inequalities in~\eqref{Eq:f_eta_ordered_BD} which makes the actual privacy loss $\geq 2\delta$, in violation of the inequality constraint  in~\eqref{eq:y_eta_condition}, which is a contradiction. 

From the discussion in the previous paragraph, for the next linear region, we now search for a group of length $\bar{\mu}$ with the smallest possible PMF, which is found to be  $(b-1)^{th}$ group, whose PMF is $\phi^1_{b-1}$. For this group to split and enter into the linear region, we should have  $\delta=\phi^1_{b-2} $. As explained in the first paragraph of this proof, the same process follows in this linear region too, i.e., for $\overline{\delta}_{1,1} < \delta\leq \underline{\delta}_{2,1}$. In the next flat region, i.e., for $\underline{\delta}_{2,1} \leq \delta\leq \overline{\delta}_{2,1}$ the step lengths, as compared to previous flat region, which have been altered are $\ell_{b-1} = (\bar{\mu}-1)$ and $\ell_{b} = \bar{\mu}$. Now, observe that $b^{th}$ group has length $\bar{\mu}$ and from the discussion in the previous paragraph, this group splits at the end of this flat region, i.e., when $\delta=\phi^2_{b-2} + \phi^2_{b-1} = \phi^2_{0} e^{-(b-2)\epsilon} (1+e^{-\epsilon}) \equiv \overline{\delta}_{2,1}$. In this case the inequality in~\eqref{Eq:f_eta_ordered_BD} are violated by both $\phi^2_{b-2}$ and $\phi^2_{b-1}$, hence the summation and the reasoning for~\eqref{eq:delta_overline_BD} (this was missing in SD neighborhood case). For the same reason, in the next linear region i.e., for $\overline{\delta}_{2,1} < \delta\leq \underline{\delta}_{2,2}$ the normalizing factor  $(1+e^{-\epsilon}) \equiv \xi_2^\epsilon$ is used while computing the PMF values in~\eqref{eq:step_height_thm}. 

These alternate flat and linear intervals are formed as $\delta$ increases and the process continues until all the remaining groups of lengths   $\bar{\mu}$ split into groups of lengths $(\bar{\mu}-1)$.
The expression for $\underline{\delta}_{h,j}$, $h \in [b-1]_+, ~j \in [h]_+$ in~\eqref{eq:delta_underline_BD} is computed using $f^\star(0)$ vs. $\delta$ curve (see Fig.~\ref{fig:BD_f0_vs_delta})  in the linear region between $\phi^{k-1}_{0}$ and  $\phi^{k}_{0}$, the slope corresponding to $\psi_0^k(\delta)$ i.e.  $e^{(b-h)\epsilon}/ \xi_j^{\epsilon}$ from \eqref{eq:step_height_thm}. 
i.e.,
\begin{align}
   \underline{\delta}_{h,j} &=  \overline{\delta}^{\epsilon}_{h,j-1} + (\phi^{k}_{0} - \phi^{k-1}_{0}) e^{-(b-h)\epsilon} \xi_j^{\epsilon} \nonumber\\
   &= \phi^{k-1}_{0} e^{-(b-h)\epsilon} \xi_j^{\epsilon} + (\phi^{k}_{0} - \phi^{k-1}_{0}) e^{-(b-h)\epsilon} \xi_j^{\epsilon} \nonumber\\
   &= \phi^{k}_{0} e^{-(b-h)\epsilon} \xi_j^{\epsilon}.
\end{align}

Now, we find the simplified expression for the value of $\phi^k_{0}$  from the fact that  
$\phi^k_{0} + \sum_{i=1}^{b+1} \ell_i \phi^k_{i} = 1$ as following:
\begin{align}
\label{eq:f_star_0_thm}
    \phi^k_{0} &\left(1 + \bar{\mu} \sum_{i=1}^{b-h} e^{-i\epsilon} \right.+ (\bar{\mu}-1)\!\!\!\sum_{\substack{i=b-h+1\\i \neq b-h+j+1}}^{b} e^{-i\epsilon} \nonumber\\ &\!\!\!\!+ \bar{\mu} u_{hj}e^{-(b-h+j+1)\epsilon} \left. + (r+h-u_{hj}) e^{-(b+1)\epsilon}\right) \!\!=\!\! 1
\end{align}
By further simplifying, we get:
\begin{align*}
\label{eq:f_star_0_thm2}
    \!\!\!\!\phi^k_{0} & = \left(1+\bar{\mu} \sum_{i=1}^{b} e^{-i\epsilon}  -\bar{\mu}(1-u_{hj}) e^{-(b-h+j+1)\epsilon}\right.\\ 
    & \!\!\!\!\!\!+ e^{-(b-h+j+1)\epsilon}\!+\!(r\!+\!h\!-\!u_{hj}) e^{-(b+1)\epsilon} \!-\! \left. \!\!\!\!\!\!\sum_{i=b-h+1}^{b}\!\!\!\!\!\!\! e^{-i\epsilon} \right)^{-1} \\
    \!\!\!\!\phi^k_{0} & = \left(\bar{\mu} \alpha^{\epsilon}_{hj} + \beta^{\epsilon}_{hj}\right)^{-1}, ~\mbox{see~\eqref{eq:phi_0}}.
\end{align*}

The last group has length $b+r$ which is less than $\bar{\mu}$ in our case. This assumption simplifies the analysis, hence tractable. Even though the pattern is same, the case $b+r \geq \bar{\mu}$ complicates the analysis because, for $j \in [h]_+$ and for every $h \in [b]_+$,  instead of increasing $h$ when $j =h$, $j$ must be increased beyond $h$ to accommodate the groups of lengths $\bar{\mu}$ created by the last group whenever its length $r+h$ exceeds $\bar{\mu}$. Hence, we do not discuss the analysis of this case, but we do  provide some numerical results in Section~\ref{Sec:Simulations}.

\section*{Declarations}

\subsection*{Availability of data and materials}
Not applicable.

\subsection*{Funding}
This research was supported by the Director, Cybersecurity, Energy Security, and Emergency Response, Cybersecurity for Energy Delivery Systems program, of the U.S. Department of Energy, under contract DE-AC02-05CH11231.  Any opinions, findings, conclusions, or recommendations expressed in this material are those of the authors and do not necessarily reflect those of the sponsors of this work.

\subsection*{Acknowledgements}
Not applicable.

\subsection*{Ethics approval and consent to participate}
Not applicable.

\subsection*{Consent for publication}
Not applicable.

\bibliographystyle{bmc-mathphys}
\bibliography{bibfile}



\section*{Note}
A preliminary version of this paper~\cite{kadam2023optimum} is published in the proceedings of IEEE SmartNets 2023 (https://smartnets.ieee.tn/). 

\end{document}